\documentclass[manuscript]{acmart}
\AtBeginDocument{%
  }

\usepackage{subcaption}
\usepackage{bbm}
\newtheorem{theorem}{Theorem}
\begin{document}

\title{Human-AI Collaboration in Cloud Security: Cognitive Hierarchy-Driven Deep Reinforcement Learning}


\author{Zahra Aref}
\affiliation{%
  \institution{Rutgers University}
  \department{Electrical and Computer Engineering}
  \city{New Brunswick}
  \state{NJ}
  \country{USA}}
\email{zahra.aref@rutgers.edu}
\orcid{0009-0002-9665-2515}

\author{Sheng Wei}
\affiliation{%
  \institution{Rutgers University}
  \department{Electrical and Computer Engineering}
  \city{New Brunswick}
  \state{NJ}
  \country{USA}}
\email{sheng.wei@rutgers.edu}
\orcid{}

\author{Narayan B. Mandayam}
\affiliation{%
  \institution{Rutgers University}
  \department{WINLAB, Electrical and Computer Engineering}
  \city{New Brunswick}
  \state{NJ}
  \country{USA}}
\email{narayan@winlab.rutgers.edu}
\orcid{}

\renewcommand{\shortauthors}{Aref et al.}

\begin{abstract}
Given the complexity of multi-tenant cloud environments and the growing need for real-time threat mitigation, Security Operations Centers (SOCs) must integrate AI-driven adaptive defense mechanisms to counter Advanced Persistent Threats (APTs) effectively. However, SOC analysts face challenges in countering adaptive adversarial tactics, requiring intelligent decision-support frameworks. 
To enhance human-AI collaboration in SOCs, we propose a Cognitive Hierarchy Theory-driven Deep Q-Network (CHT-DQN) framework that models the interactive decision-making process between SOC analysts and AI-driven APT bots. Our approach assumes that the SOC analyst (defender) operates at cognitive level-1, anticipating attacker strategies, while the APT bot (attacker) follows a level-0 policy, making naive exploitative decisions. By incorporating CHT into DQN, the framework enhances adaptive SOC defense strategies using Attack Graph (AG)-based reinforcement learning.
Through extensive simulation experiments across varying attack graph complexities, the proposed model consistently achieves higher average weighted data protection and reduced action discrepancies, particularly when SOC analysts utilize CHT-DQN. A theoretical lower bound analysis further validates the superiority of CHT-DQN over standard DQN, demonstrating that as AG complexity increases, CHT-DQN yields higher Q-value performance. A comprehensive human-in-the-loop (HITL) evaluation involving IT security experts was conducted on Amazon Mechanical Turk (MTurk). The findings indicate that SOC analysts who employ CHT-DQN-driven transition probabilities demonstrate greater alignment with adaptive DQN attackers, thereby achieving enhanced data protection. Moreover, human decision patterns align with Prospect Theory (PT) and Cumulative Prospect Theory (CPT): participants tend to avoid reselection after failure but persist after success. This asymmetry reflects heightened loss sensitivity and biased probability weighting—underestimating gains after loss and overestimating continued success.
The findings highlight the significant potential of integrating cognitive modeling into deep reinforcement learning to improve Security Operations Center functions. This combination presents a promising avenue for developing real-time, adaptive mechanisms for cloud security.
\end{abstract}




\maketitle

\section{Introduction}

Security Operations Centers (SOCs) are the backbone of modern cybersecurity, safeguarding cloud infrastructures against Advanced Persistent Threats (APTs) by integrating threat detection, incident response, and continuous monitoring across multi-cloud and multi-tenant environments. They employ Security Information and Event Management (SIEM), machine learning, and automated threat-hunting tools to detect cyber threats, analyze attack patterns, and identify system vulnerabilities~\cite{chhetri2024humanAI}. However, SOCs face increasing operational challenges, including overwhelming alert volumes, compliance complexities, vendor dependencies, and the rising sophistication of AI-driven adversarial threats\cite{schmidt2023distributedSOC}. These challenges are exacerbated in cloud environments, where the attack surface expands dynamically due to virtualized resources, shared infrastructures, and evolving cyber threat landscapes\cite{muniz2021modern, peiris2021threat}.

Despite advancements in AI-driven security tools, SOC analysts continue to struggle with high false positive rates, cognitive fatigue, and a shortage of skilled cybersecurity professionals, limiting their ability to respond efficiently to complex attacks\cite{chhetri2023riskaware}. Prior research has demonstrated that subjectivity in human decision-making can significantly impact the effectiveness of DRL-based security defenses, particularly in cloud storage environments\cite{aref2022impact}. To mitigate these challenges, AI-augmented SOC frameworks must integrate human cognitive insights with machine-learning-driven adaptive defenses, enabling real-time decision-making and proactive adversarial reasoning\cite{chhetri2022alert, tariq2024A2C}. This human-AI collaboration is essential to enhancing SOC operations, as it combines human intuition and strategic thinking with AI’s ability to analyze vast amounts of threat intelligence data in real-time~\cite{prebot2023learning, zimmermann2019moving}.

This work incorporates Cognitive Hierarchy Theory (CHT) into Deep Reinforcement Learning (DRL) to model adversarial strategies in cloud security defense. Unlike traditional game-theoretic equilibrium-based approaches, which assume perfect rationality, CHT models adversaries at different cognitive levels, better reflecting the bounded rationality seen in real-world cyber conflicts~\cite{camerer2004cognitive, kleiman2016coordinate}. In our framework, the SOC analyst (defender) at cognitive level-1 anticipates APT attacker strategies at level-0, where the attacker employs exploitative decision-making but lacks deeper strategic foresight. By infusing CHT into DQN, we enable SOC analysts to predict attacker behaviors and dynamically adjust defense mechanisms in real time, enhancing threat anticipation and mitigation.

To support structured adversarial modeling, we integrate AGs, a widely used method in cyber threat intelligence, to map attack progression, visualize vulnerabilities, and optimize SOC decision-making\cite{wang2008attack, nadeem2021sage, ibrahim2019attack}. While AGs are traditionally static, their integration with CHT-driven DRL allows for dynamic threat modeling, equipping SOC analysts with real-time attack prediction capabilities and adaptive countermeasure optimization\cite{xrl2023userstudy}.

Additionally, Human-in-the-Loop (HITL) models have proven effective in SOCs, as they align AI-driven threat detection with human decision-making heuristics, improving response efficiency and security posture\cite{crabtree2019cyber, gratian2018correlating}. HITL enhances SOC analyst trust in AI systems, reduces false positives, and facilitates adaptive learning based on human feedback\cite{zimmermann2019moving, katakwar2023attackers}. By embedding HITL into our CHT-augmented DRL framework, we create a human-AI collaboration paradigm, ensuring that SOC analysts and AI agents work together effectively in cloud security operations.

Our research contributions are summarized as follows:
\begin{itemize}
\item We introduce a human-AI collaborative framework that leverages Attack Graphs (AGs) to provide SOC analysts with a visual and analytical representation of cloud infrastructure vulnerabilities. This framework incorporates a human-in-the-loop (HITL) approach to facilitate real-time, adaptive defense strategies against evolving cyber threats.
\item Our model employs Deep Reinforcement Learning (DRL) to simulate a multi-agent interaction between an AI-driven attacker and a human-AI SOC analyst, enhancing real-time decision-making in SOC operations against Advanced Persistent Threats (APTs).
\item By integrating Cognitive Hierarchy Theory (CHT) into Deep Q-Networks (DQN), we develop a hierarchical reasoning framework that models multi-level adversarial strategies, improving SOC analysts' ability to anticipate and counter sophisticated attack behaviors. This framework is evaluated in four distinct scenarios, where the SOC analyst employs CHT-DQN and DQN strategies against the attacker using DQN and random selection.
\item We propose and validate a theoretical performance guarantee, proving via induction that as the AG size increases, the SOC analyst’s Q-value under the CHT-DQN policy consistently outperforms the standard DQN policy, assuming a stationary attack strategy. This result formalizes the robustness of cognitive-aware AI models in cloud security defense.
\item Our study includes two human-interactive, web-based DRL experiments on Amazon Mechanical Turk (MTurk), where human SOC analysts defend cloud infrastructure against AI-driven DQN attackers. By comparing reward-based decision-making with CHT-DQN-driven transition probability insights, we validate the practical benefits of human-AI collaboration in SOCs for real-time threat mitigation.
\end{itemize}

The remainder of this paper is organized as follows. Section 2 reviews related work on deep reinforcement learning for SOC defense, attack graphs in cybersecurity, cognitive models in cyber defense, and human-AI collaboration in SOC operations. Section 3 provides the necessary background on deep Q-learning, CHT, and attack graphs as applied to SOC decision-making. Section 4 introduces our proposed CHT-driven DQN framework, detailing its integration with attack graphs for adaptive SOC defense. Section 5 describes our experimental setup, including both simulation-based evaluations and human-in-the-loop experiments conducted on Amazon Mechanical Turk (MTurk). Section 6 presents experimental results, comparing the performance of CHT-DQN with standard DQN in SOC environments and analyzing human decision-making behaviors in cybersecurity contexts. Finally, Section 7 concludes the paper with key findings, implications for AI-augmented SOCs, and directions for future research in cloud security.
\section{Related Work}

Cybersecurity research increasingly employs advanced computational models to address the growing complexity of cyber threats, particularly APTs in cloud environments. Modern SOCs integrate machine learning (ML), automated threat-hunting, and Security Information and Event Management (SIEM) tools to detect and mitigate threats in real-time~\cite{peiris2021threat, muniz2021modern}. However, challenges such as alert fatigue, vendor dependencies, compliance issues, and the increasing sophistication of AI-driven adversaries necessitate more intelligent, proactive security frameworks.

Recent studies highlight the importance of human-AI collaboration in SOC operations, where AI-powered decision-support systems complement human analysts in managing cyber threats efficiently~\cite{chhetri2024humanAI}. This has led to the exploration of Deep Reinforcement Learning (DRL) for adaptive defense, AGs for attack pathway visualization, CHT for adversarial reasoning, and HITL strategies for SOC decision-making. This section reviews key research in these areas, positioning our CHT-driven DRL framework within the existing cybersecurity landscape.

\subsection{Deep Reinforcement Learning for SOCs and APT Defense}
Deep Reinforcement Learning (DRL) has emerged as a powerful approach for cyber threat detection and mitigation in SOCs. DRL enables adaptive defense strategies, where an AI agent continuously learns from adversarial interactions to improve its decision-making.
The survey by Nguyen et al.\cite{nguyen2021deep} explores DRL’s applicability to cyber-physical systems, intrusion detection, and multi-agent security frameworks, discussing various DRL approaches such as value-based methods (DQN), policy gradients, and actor-critic frameworks. Other studies apply DRL to real-time cloud security, demonstrating its potential for resource allocation and adaptive defense against APTs\cite{min2018defense}. Beyond cybersecurity, DRL has also been explored for optimizing design verification workflows, where RL enables efficient test generation and automated verification of complex hardware and software systems\cite{aref2024advanced}. These applications highlight DRL’s versatility in optimizing decision-making under uncertainty, reinforcing its suitability for adaptive cyber defense.
More specific to cloud-based SOCs, the work in~\cite{sakthivelu2023adaptive} introduces a DRL-based CPU allocation model for cloud storage defense, using a Colonel Blotto game framework to counter APT attacks. While effective in storage optimization, it does not consider broader threat adaptation across SOC environments.

Additionally, Chhetri et al.~\cite{chhetri2023riskaware} propose a risk-aware AI model for cloud cost optimization, incorporating dynamic threat assessment mechanisms to manage cloud security risks effectively. This aligns with our work, which enhances SOC defenses by integrating CHT with DRL to model adversarial reasoning and proactive threat mitigation in cloud infrastructures.

\subsection{Attack Graphs in Cybersecurity and Cloud Defense}
AGs can play a crucial role in SOC threat modeling, helping analysts visualize, analyze, and predict the movement of adversaries within cloud infrastructures. AGs provide a structured representation of cyber threats, allowing security teams to prioritize vulnerabilities and allocate resources strategically~\cite{wang2008attack, nadeem2021sage, ibrahim2019attack}.
Several studies explore AG-based threat intelligence. For example, Sheyner et al.\cite{sheyner2002automated} automate AG generation for network vulnerability analysis, while Wu et al.\cite{wu2022research} use AGs to evaluate power communication networks, prioritizing risks based on security metrics. Additionally, Sengupta et al.~\cite{sengupta2019general} propose an AG-driven Markov game framework for APT detection in cloud networks, leveraging Common Vulnerability Scoring System (CVSS) metrics for dynamic threat modeling.
In multi-cloud landscapes, Schmidt et al.~\cite{schmidt2023distributedSOC} investigate distributed SOC architectures, emphasizing AG-based threat modeling and adaptive defense planning. However, existing AG-based models focus primarily on reactive threat detection rather than proactive, AI-driven defense mechanisms. Our framework extends AG applications by integrating them with CHT and DRL, enabling real-time SOC decision-making against evolving APTs.

\subsection{Cognitive Models in Cyber Defense}
CHT and behavioral game theory have gained traction in cybersecurity research, providing insights into human-like decision-making in adversarial settings~\cite{camerer2004cognitive, kleiman2016coordinate}. Unlike traditional Nash equilibrium-based models, which assume perfect rationality, CHT recognizes that attackers and defenders operate at different levels of cognitive reasoning, leading to non-equilibrium decision patterns.

In cloud security, Aref et al.\cite{aref2022impact} apply Prospect Theory to model defender decision-making, demonstrating that subjective risk perception can enhance defense strategies in resource-constrained cloud environments. Similarly, Xiao et al.\cite{xiao2018attacker} formulate an APT detection game, showing that attackers exhibit risk-seeking behaviors when facing aggressive defense mechanisms.
CHT is also applied in cyber-physical security, where Mavridis et al.\cite{mavridis2023attack} use level-k reasoning to anticipate APT behaviors, while Sanjab et al.\cite{sanjab2020game} analyze UAV cybersecurity through cumulative Prospect Theory, capturing adversarial risk dynamics.

Our research builds upon these models by integrating CHT with DRL in a cloud SOC framework, allowing SOC analysts to predict APT strategies and dynamically adjust defenses based on multi-level adversarial cognition.

\subsection{Human-AI Collaboration and Human-in-the-Loop SOCs}
Human-AI collaboration (HITL) is essential for enhancing SOC efficiency, as AI-driven automation alone cannot fully replace human expertise in cybersecurity operations. AI models must align with human cognitive processes to improve SOC decision-making, reduce false positives, and enhance analyst interpretability~\cite{tariq2024A2C, chhetri2024humanAI}.
Several studies highlight HITL in SOCs. Chhetri et al.\cite{chhetri2022alert} propose an AI-driven alert prioritization system to reduce cognitive overload in SOC analysts, improving threat response efficiency. Furthermore, Tariq et al.\cite{tariq2024A2C} introduce a modular AI-human decision framework that optimizes SOC workflow automation while ensuring human oversight in high-risk decision-making.

Explainable AI (XAI) is another key area in human-AI collaboration. A recent study explores how explainable reinforcement learning models improve SOC analyst trust and decision accuracy in real-time security operations~\cite{xrl2023userstudy}. Collaborative AI learning frameworks, such as the survey in \cite{collaborative2023survey}, further, emphasize the role of human-AI teaming in security environments.
By integrating HITL with our CHT-DQN model, we enable SOC analysts to interact with AI-driven security tools in real time, leveraging human cognitive intuition alongside machine-learning automation for adaptive cloud defense strategies.
\section{Preliminaries}

This study integrates human computation principles into cloud security to defend against APTs using Cognitive Hierarchy Theory (CHT) in Deep Reinforcement Learning (DRL).

\subsection{Deep Q-Learning}
Deep Q-Learning approximates the Q-function, which estimates the expected cumulative reward of taking an action in a given state under policy $\pi$. The environment is modeled as a Markov Decision Process (MDP) defined by $\mathcal{M} = (S, A, T, \gamma, R)$, where $S$ is the state space, $A$ is the action space, $T$ represents transition probabilities, $\gamma$ is the discount factor, and $R$ is the reward function. The goal is to learn an optimal policy $\pi^*$ that maximizes the expected discounted return:
\begin{equation}
    Q_{\pi}(s, a;\boldsymbol{\theta}) = \mathbb{E}_{\pi}\left[\sum_{k=0}^{\infty} \gamma^k R(s_k, a_k, s_{k+1}) \mid s_0 = s, a_0 = a\right].
\end{equation}
Deep Q-Networks (DQNs) approximate $Q^*(s, a; \boldsymbol{\theta})$ using a neural network with parameters $\boldsymbol{\theta}$, trained iteratively to minimize the Bellman error:
\begin{equation}
    L(\boldsymbol{\theta}) = \mathbb{E} \Big[(y - Q(s, a; \boldsymbol{\theta}))^2 \Big],
\end{equation}
where the target value $y$ is defined as $y = R(s, a, s') + \gamma \max_{a'} Q(s', a'; \boldsymbol{\theta^-})$, and $\boldsymbol{\theta^-}$ represents a delayed copy of the Q-network parameters for stable updates.

\subsection{Cognitive Hierarchy Theory and K-Level Reasoning}
Cognitive Hierarchy Theory (CHT) models strategic reasoning by assuming that decision-makers (agents) operate at different levels of thinking, where a level-$k$ agent optimizes its decisions based on the expected behavior of agents at levels $0$ to $k-1$.

\textbf{Definition 1 (Level-0 Policy).} Level-0 agents select actions randomly or based on simple heuristics.

\textbf{Definition 2 (Level-$k$ Policy).} A level-$k$ agent assumes that opponents reason up to level $k-1$ and selects the best response accordingly.

For a player $i$ with strategy set $m_i$, the expected utility when selecting strategy $s_i^j$ at level-$k$ is given by:
\begin{equation}
    E_k(\pi_i(s_i^j)) = \sum_{j'=1}^{m_{-i}} \pi_i(s_i^j, s_{-i}^{j'}) \sum_{h=0}^{k-1} g_k(h) \cdot P_h(s_{-i}^{j'}),
    \label{eq:expected_payoff}
\end{equation}
where $\pi_i(s_i^j, s_{-i}^{j'})$ represents the payoff of player $i$ given opponent state $s_{-i}^{j'}$. The probability $P_h(s_{-i}^{j'})$ captures the likelihood of level-$h$ opponents selecting strategy $s_{-i}^{j'}$, and $g_k(h)$ is a Poisson-distributed weighting function. The best response strategy satisfies:
\begin{equation}
    P_k(s_i^\ast) = \mathbbm{1}\left(s_i^\ast = \arg\max_{s_i^j} E_k(\pi_i(s_i^j))\right).
\end{equation}

\begin{figure}[t]
    \centering
    \begin{subfigure}[b]{0.4\textwidth}
        \includegraphics[width=\textwidth]{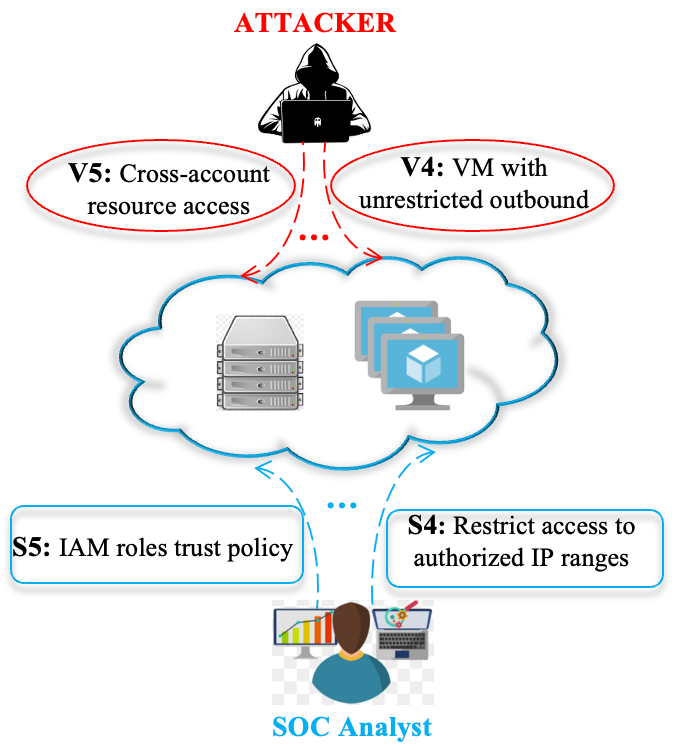}
        \caption{\textbf{AI-Augmented SOC Framework.} The attacker exploits vulnerabilities (V4, V5), while the SOC analyst deploys defensive strategies (S4, S5) to mitigate threats.}
        \label{fig:system_model}
    \end{subfigure}
    \hfill
    \begin{subfigure}[b]{0.4\textwidth}
        \includegraphics[width=\textwidth]{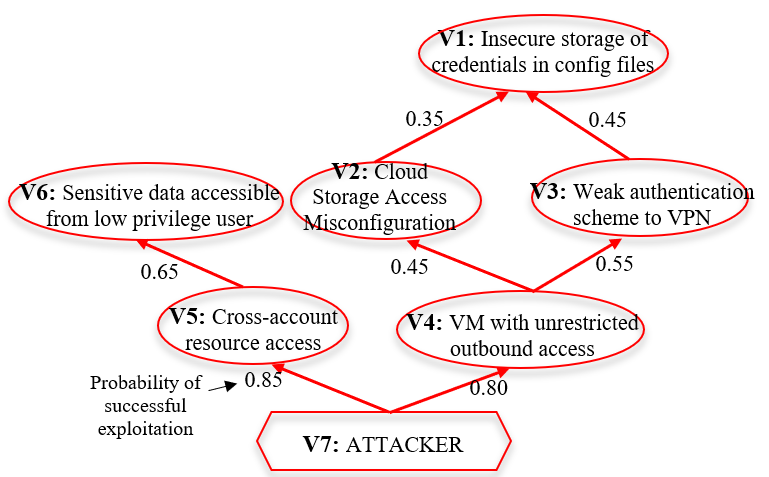} 
        \caption{\textbf{Example of a Cloud Attack Graph.} This diagram illustrates vulnerabilities (V1-V6) and probabilistic exploit paths within a cloud infrastructure.}
        \label{fig:AG}
    \end{subfigure}
    \caption{Illustrations of real-time SOC decision-making using attack graphs (AGs).}
\end{figure}
\section{CHT-driven DQN Cloud Data Defense}
In the domain of Security Operations Centers (SOCs) for cloud security, the convergence of human-AI collaboration, cognitive modeling, and machine learning presents a promising avenue for real-time threat mitigation and adaptive defense strategies. This paper introduces a novel defensive algorithm that infuses Cognitive Hierarchy Theory (CHT) into Deep Q-Learning (DQN), enabling SOC analysts to make real-time, risk-aware decisions against evolving cyber threats.

The proposed CHT-DQN algorithm operates within a stochastic game framework 
$(S, A_\mathcal{D}, A_\mathcal{A}, T, r_{\mathcal{D}}, r_{\mathcal{A}})$, 
characterized by a state space $S$, action spaces $A_{\mathcal{D}}$ and $A_{\mathcal{A}}$ 
for the SOC analyst (defender) and APT attacker, and a state-transition function: $T$.
This function, $T: S \times A_{\mathcal{D}} \times A_{\mathcal{A}} \to \Delta(S)$, 
determines the transition probabilities to a new state $s'$, given the real-time 
defensive actions taken by the SOC analyst ($a_\mathcal{D}$) and the attacker ($a_\mathcal{A}$):
\begin{equation}
    T(s, a_{\mathcal{D}}, a_{\mathcal{A}}, s') = \mathbb{P}(s' | s, a_{\mathcal{D}}, a_{\mathcal{A}}).
\end{equation}
The reward functions $r_{\mathcal{D}}: S \times A_{\mathcal{D}} \times S \to \mathbb{R}$ 
and $r_{\mathcal{A}}: S \times A_{\mathcal{A}} \times S \to \mathbb{R}$ define 
the respective utility structures for the SOC analyst and the attacker.
These reward formulations integrate cognitive modeling and risk-sensitive decision-making 
to support real-time security strategies in SOC environments.
The discount factor $\gamma \in [0,1]$ determines the relative importance of future rewards,
balancing immediate response actions with long-term security optimization.

Figure~\ref{fig:system_model} presents an AI-augmented SOC framework utilizing Attack Graphs (AGs) for real-time security operations. The attacker, depicted at the top, attempts to exploit cloud infrastructure vulnerabilities such as V4: Virtual Machine (VM) with unrestricted outbound access and V5: Cross-account resource access to compromise cloud resources. Meanwhile, the SOC analyst at the bottom, supported by Cognitive Hierarchy Theory-driven Deep Q-Network (CHT-DQN), dynamically implements countermeasures like S4: Restricting access to authorized IP ranges and S5: IAM roles trust policy to defend against security breaches. This interaction models the adversarial dynamics in cloud security, allowing SOCs to optimize resource allocation and enhance real-time threat mitigation.

\subsection{Attack Graph}
The Attack Graph (AG) is central to SOC-driven decision-making, offering a structured representation of network vulnerabilities and potential attack trajectories~\cite{wang2008attack,nadeem2021sage,ibrahim2019attack,chu2010visualizing, poolsappasit2011dynamic}. 
The SOC analyst can minimize security risks and get equipped for potential cyber-attacks in advance by evaluating the AG. 
Within an AI-augmented SOC, attack graphs serve as real-time intelligence tools, assisting SOC analysts in:
\begin{itemize}
\item Anticipating potential threats based on historical attack data.
\item Prioritizing incident response by estimating the likelihood of successful attacks.
\item Improving real-time security decision-making through automated risk quantification and adaptive threat modeling.
\end{itemize}
By leveraging the attack graph, SOC analysts can preemptively adjust security controls, dynamically deploy countermeasures, and minimize attack impact in real-time.

An attack graph~\cite{husak2018survey}, $\mathcal{G}$, is formally defined as the tuple:
\begin{equation}
    \mathcal{G} = (\mathcal{V}, \mathcal{E}, \mathbb{P}),
\end{equation}
where:
\begin{itemize}
    \item $\mathcal{V} = \{v_1, v_2, ..., v_N\}$ is the set of nodes, representing vulnerabilities or threat attributes in cloud environments.
    \item $\mathcal{E} \subseteq \mathcal{V} \times \mathcal{V}$ is the set of directed edges, termed exploits, which define possible attack paths.
    \item $\mathbb{P} = \{\mathbb{P}(v) \mid v \in \mathcal{V}\}$ is the set of exploitation probabilities, where each $\mathbb{P}(v) \in [0,1]$ quantifies the likelihood of successful exploitation of vulnerability $v$.
\end{itemize}

In our model, the attack graph $\mathcal{G}$ is dynamically updated at fixed intervals using a Maximum Likelihood Estimation (MLE) approach:
\begin{equation}
    \mathbb{P}(v) = \frac{\phi(v)}{\sum_{v' \in \mathcal{V}} \phi(v')},
\end{equation}
where $\phi(v)$ denotes the observed exploitation frequency of node $v$ over time.

This ensures that exploitation probabilities $\mathbb{P}(v)$ reflect the most recent attack-defense interactions. 
By continuously updating $\mathbb{P}(v)$, our framework dynamically captures adversarial behaviors, 
allowing SOC analysts to make data-informed security decisions based on evolving attack trends.
This real-time adaptability enhances the evaluation of defensive reinforcement learning algorithms in a controlled, yet dynamic, cyber threat landscape.

\subsection{Model Dynamics and Utility Functions}
We model the SOC environment as a stochastic game where each cloud system node $i$ in the attack graph (AG) has a binary security state representing its compromise status. Let  $s^t = [s_1^t, s_2^t, ..., s_N^t]^\top$  be the security state vector at time $t$, where $s_i^t \sim \text{Bern}(p_i)$ indicates whether node $i$ is compromised ($s_i^t = 1$) or secure ($s_i^t = 0$).

The AG representation of the cloud system is illustrated in Figure~\ref{fig:AG}. The hexagonal node represents the attacker, while oval nodes represent cloud system vulnerabilities. The attacker node (7) serves as the parent of vulnerabilities (4) and (5), indicating possible exploit paths.

To enable real-time SOC decision-making, our AG  quantifies exploitation likelihoods, allowing adaptive SOC response strategies. Given security mechanisms—including firewalls, access control policies, and cryptographic techniques—SOC analysts dynamically allocate defense resources to optimize security while minimizing operational costs.

\textbf{State-Action Representation:}
At each time step $t$, the APT attacker selects a node  $a_{\mathcal{A}}^t$  from the attack graph $\mathcal{G}$ to exploit, where:
$a_{\mathcal{A}}^t \in \mathcal{G}$. Each vulnerability in the cloud infrastructure corresponds to an AG node, representing:
\begin{itemize}
    \item Weak authentication $\&$ access control
    \item Security misconfiguration $\&$ software vulnerabilities
    \item Network-based threats (e.g., botnets, spoofing, VM co-location)
    \item Storage-based attacks $\&$ application exploits~\cite{rakotondravony2017classifying, ullah2018data}.
\end{itemize}

Conversely, the SOC analyst, leveraging  CHT-driven DQN, selects a node $a_{\mathcal{D}}^t$  to protect: $a_{\mathcal{D}}^t \in \mathcal{G}$.
Defensive strategies include:
\begin{itemize}
    \item Intrusion detection $\&$ cloud-based antivirus~\cite{coppolino2017cloud}
    \item Cryptographic authentication $\&$ VM integrity monitoring~\cite{rakotondravony2017classifying}
    \item Storage attack mitigation via security monitoring~\cite{tari2014security}
    \item Binary code analysis for exploit prevention~\cite{lee2017cloudrps}.
\end{itemize}

\textbf{State Transition Model:}
The next  cloud security state  is determined by:
\begin{equation}
    s^{t+1} = \left\lbrace[1 - \delta^{t}_i]_{1 \leq i \leq N} \right\rbrace,
\end{equation}
where  $\delta^t_i$  represents whether node $i$ was compromised at time step $t$:
\begin{equation}
    \delta^t_i (a_{\mathcal{A},i}, a_{\mathcal{D},i}) = \left\{
    \begin{array}{ll}
        0, & \text{if } a^t_{\mathcal{A},i} = 1 \text{ and } a^t_{\mathcal{D},i} = 0 \\
        1, & \text{otherwise}
    \end{array}
    \right.
\end{equation}
The attack graph probabilities $\mathbb{P}(s'|s,a_{\mathcal{A}}, a_{\mathcal{D}})$ are updated dynamically based on the success rate of prior attacks. This ensures SOC analysts continuously refine their defense strategies based on observed historical attack patterns.

\textbf{Utility Functions (Reward Structures):}
To model the cost-benefit trade-offs in SOC operations, we define utility functions for both the  APT attacker and  SOC analyst.

The attacker's utility function at time $t$ considers the estimated  data exfiltration  and the  cost of attack :
\begin{equation}
    u^t_{\mathcal{A}} = \sum_{i=1}^{N} \bigg[ 
        \big( \alpha_{\mathcal{A},i} \hat{b}_i - \beta_{\mathcal{A},i} c_{\mathcal{A},i} \big) \cdot \Big(1 - \delta^t_i \Big)
    - \big( \alpha_{\mathcal{A},i} \hat{b}_i + \beta_{\mathcal{A},i} c_{\mathcal{A},i} \big) \cdot \delta^t_i \bigg].
    \label{eq:u_a}
\end{equation}
where $\hat{b}_i$ represents the attacker's estimated data volume at node $i$, while $c_{\mathcal{A},i}$ denotes the computational cost of exploiting the vulnerability. The parameter $\alpha_{\mathcal{A},i}$ quantifies the severity of data exfiltration, and $\beta_{\mathcal{A},i}$ controls the trade-off between attack success probability and resource expenditure. The attacker aims to exploit nodes with high expected data payoff while adapting strategies to SOC defenses.

The SOC analyst optimizes their defense policy $\pi_{\mathcal{D}}$ to maximize expected cloud security while minimizing defense expenditures. The reward function at time $t$ is defined as:
\begin{equation} u^t_{\mathcal{D}} = \sum_{i=1}^{N} \Big[ \big( \alpha_{\mathcal{D},i} b_i - \beta_{\mathcal{D},i} c_{\mathcal{D},i} \big) \cdot \delta^t_i - \big( \alpha_{\mathcal{D},i} b_i + \beta_{\mathcal{D},i} c_{\mathcal{D},i} \big) \cdot (1 - \delta^t_i) \Big]. \label{eq:u_d} \end{equation}
where $b_i$ represents the actual data stored at node $i$, and $c_{\mathcal{D},i}$ denotes the cost of securing node $i$. The parameters $\alpha_{\mathcal{D},i}$ and $\beta_{\mathcal{D},i}$ serve as trade-off coefficients balancing data protection benefits against resource expenditures. The SOC analyst strategically selects actions $a^t_{\mathcal{D}}$ to maximize long-term security gains by adjusting defenses based on real-time attack graphs and adversarial behaviors.

\begin{figure}[htbp]
    \centering
    \includegraphics[width=0.7\linewidth]{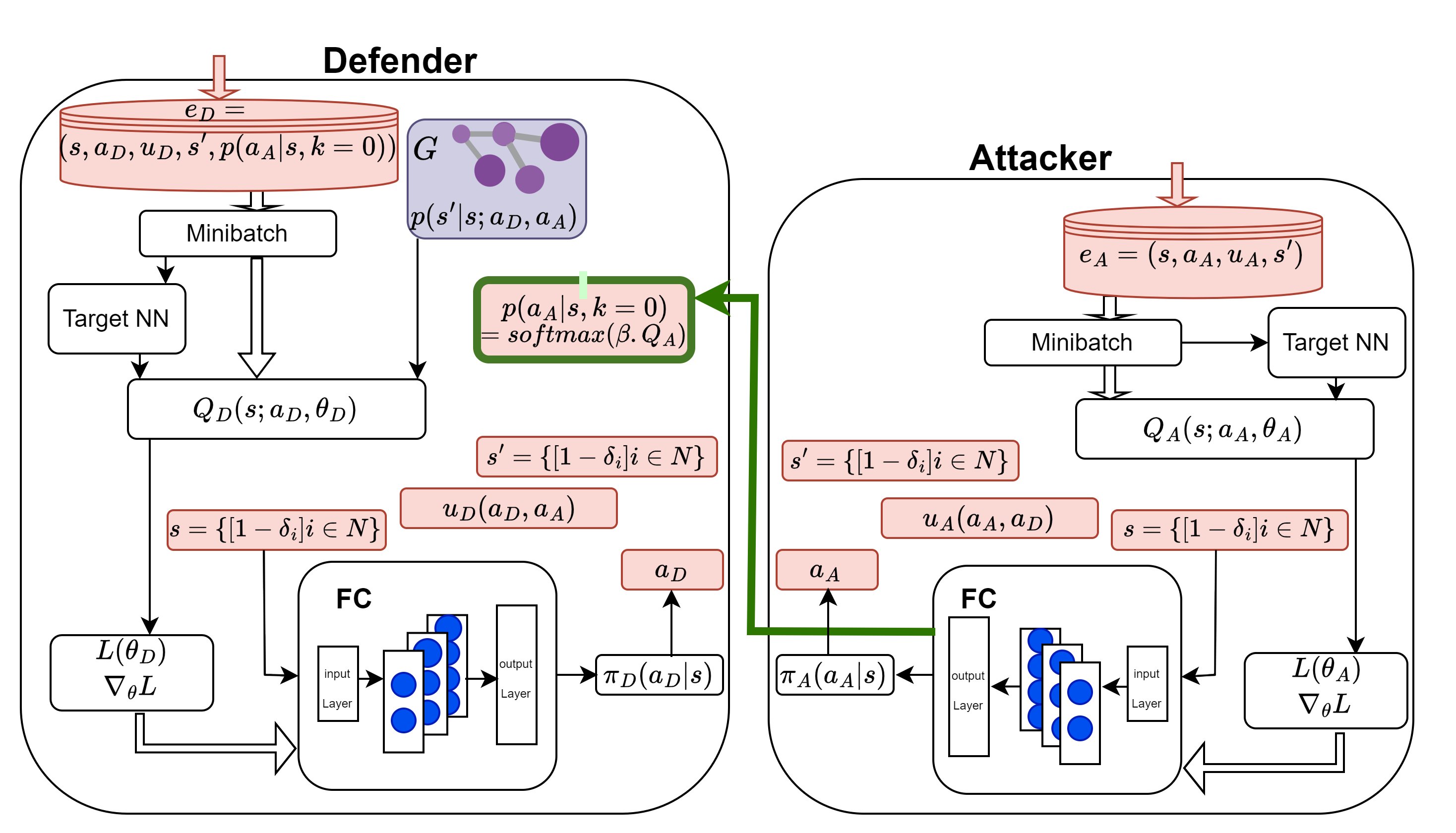}
    \caption{Overview of the proposed CHT-DQN framework for cloud security. 
    The framework models interactions between the SOC analyst (defender) and the APT attacker using AGs and deep reinforcement learning.
    The SOC analyst module includes experience replay storage ($e_{\mathcal{D}}$), minibatch sampling, and a target neural network (NN) for updating policy $\pi_{\mathcal{D}}(a_{\mathcal{D}}|s)$. 
    The SOC analyst computes $Q_{\mathcal{D}}(s, a_{\mathcal{D}}, \boldsymbol{\theta}_{\mathcal{D}})$ and applies a softmax function on $Q_{\mathcal{A}}$ to estimate the attacker's action probabilities $\mathbb{P}(a_{\mathcal{A}} | s, k = 0)$. 
    Similarly, the attacker module stores experiences ($e_{\mathcal{A}}$), samples minibatches, and optimizes its policy $\pi_{\mathcal{A}}(a_{\mathcal{A}} | s)$. 
    The forward connections highlight the SOC analyst’s predictive modeling of the attacker’s actions, leveraging the Cognitive Hierarchy Theory (CHT). 
    Both modules use fully connected (FC) layers to process state-action pairs and optimize loss functions $L_{\mathcal{D}}(\boldsymbol{\theta}_{\mathcal{D}})$ and $L_{\mathcal{A}}(\boldsymbol{\theta}_{\mathcal{A}})$.}
    \label{fig:model}
\end{figure}

\subsection{CHT-DQN Algorithm}

This stochastic game models the Security Operations Center (SOC) defense against cloud-based cyber threats, where the SOC analyst and attacker interact dynamically in a sequential decision-making environment. Each agent maximizes its utility under the assumption that the other agent is doing the same. We follow the approach proposed in \cite{kleiman2016coordinate}, which extends Cognitive Hierarchy Theory (CHT) from single-action games \cite{camerer2004cognitive} to sequential decision processes. Level-$K$ agents respond optimally to level-$(K-1)$ agents, establishing a structured reasoning hierarchy. This enables SOC analysts to anticipate adversarial strategies in real-time security operations.

Our approach integrates CHT into DQN, enabling SOC analysts to make real-time adaptive decisions against evolving cyber threats. By modeling a hierarchical cognitive structure, the SOC analyst (level-$1$) anticipates the attacker's (level-$0$) behavior. This structured decision-making process enhances SOC efficiency in alert prioritization, incident response, and cyber defense. Figure~\ref{fig:model} illustrates our proposed learning framework.

\paragraph{CHT-DQN Formulation.}
In our formulation, we extend CHT principles as defined in Equation~\eqref{eq:expected_payoff}, following \cite{kleiman2016coordinate}. Given that the SOC analyst operates at level-1 reasoning, it assumes the attacker follows a level-0 policy, which lacks strategic adaptation. The SOC analyst’s real-time defensive actions are optimized based on attack graph insights and learned attack patterns.

The expected payoff function \( E_k(\pi_i(s_i^j)) \) corresponds to the SOC analyst’s Q-value function, given by:
\begin{equation}
    Q^1_{\mathcal{D}}(s, a_{\mathcal{D}}; \boldsymbol{\theta}_{\mathcal{D}}) = \sum_{s'} \mathbb{P}(s'|s,a_\mathcal{D}) \big[u_{\mathcal{D}}(s, a_{\mathcal{D}}, s')
     + \gamma \max_{a'_{\mathcal{D}}} Q^1_{\mathcal{D}}(s', a'_{\mathcal{D}}; \boldsymbol{\theta}_{\mathcal{D}}) \big],
\end{equation}
where \(\mathbb{P}(s' | s, a_{\mathcal{D}})\) incorporates the attacker’s policy, modeling the transition probability structure from the attack graph.

\paragraph{Q-value Function of Attacker.} The level-0 agent, which is the cloud attacker,  does not use strategic thinking and attempts to reach its goal without considering how the other player might affect its action.
For the attacker at level-0:
\begin{equation}
\label{eq:Q_A}
    Q^0_{\mathcal{A}}(s, a_{\mathcal{A}}; \theta_{\mathcal{A}}) = \sum_{s'} p(s' | s, a_{\mathcal{A}}) \big[u_{\mathcal{A}}(s, a_{\mathcal{A}}, s')
    + \gamma \max_{a'_{\mathcal{A}}} Q^0_{\mathcal{A}}(s', a'_{\mathcal{A}}; \theta_{\mathcal{A}}) \big],
\end{equation}
where \(Q^0_{\mathcal{A}}\) represents the Q-value for the attacker at level-0, \(p(s'|s,a_\mathcal{A})\) denotes the transition probabilities of the attacker, \(u_{\mathcal{A}}\) is the utility (reward) of the attacker, and \(\gamma\) is the discount factor.

\paragraph{Attacker's Policy.}
The attacker, modeled as a level-0 agent, follows an exploitation-based policy based on Q-values:
\begin{equation}
    \pi^0_{\mathcal{A}}(a_{\mathcal{A}} | s) \propto \exp(\beta Q^0_{\mathcal{A}}(s, a_{\mathcal{A}}; \boldsymbol{\theta}_{\mathcal{A}})),
\end{equation}
where \(\beta\) is an inverse temperature parameter that controls stochasticity in decision-making.

\paragraph{SOC Analyst’s Policy.}
At level-1, the SOC analyst’s policy adapts based on the attacker's expected actions, following:
\begin{equation}
    \pi^1_{\mathcal{D}}(a_{\mathcal{D}} | s) \propto \exp(\beta Q^1_{\mathcal{D}}(s, a_{\mathcal{D}}; \boldsymbol{\theta}_{\mathcal{D}})).
\end{equation}

\paragraph{CHT-Driven Transition Probabilities.}
At level-1, the SOC analyst incorporates the attacker’s level-0 policy into the transition probability:
\begin{equation}
    \mathbb{P}(s' | s, a_{\mathcal{D}}) = \sum_{a_{\mathcal{A}}} \mathbb{P}(s' | s, a_{\mathcal{D}}, a_{\mathcal{A}}) \pi^0_{\mathcal{A}}(a_{\mathcal{A}} | s),
    \label{eq:trans_prob}
\end{equation}
where \(\mathbb{P}(s' | s, a_{\mathcal{D}}, a_{\mathcal{A}})\) is obtained from the attack graph \(\mathcal{G}\). This formulation allows SOC analysts to dynamically refine their defenses based on real-time adversarial behavior.

\paragraph{Training the CHT-DQN.}
The Q-functions \( Q^0_{\mathcal{A}}(s, a_{\mathcal{A}}; \boldsymbol{\theta}_{\mathcal{A}}) \) and \( Q^1_{\mathcal{D}}(s, a_{\mathcal{D}}; \boldsymbol{\theta}_{\mathcal{D}}) \) are parameterized by \(\boldsymbol{\theta}_{\mathcal{A}}\) and \(\boldsymbol{\theta}_{\mathcal{D}}\), respectively. These parameters are learned using DQN training, where the loss function is:
\begin{equation}
   L_{\mathcal{A}}(\boldsymbol{\theta}_{\mathcal{A}}) = \mathbb{E}_{s, a_{\mathcal{A}} \sim \rho(\cdot)} \left[ \left( y_{\mathcal{A}} - Q^0_{\mathcal{A}}(s, a_{\mathcal{A}}; \boldsymbol{\theta}_{\mathcal{A}}) \right)^2 \right],
\end{equation}
where \( y_{\mathcal{A}} = u_{\mathcal{A}} + \gamma \max_{a'} Q^0_{\mathcal{A}}(s', a'; \boldsymbol{\theta}_{\mathcal{A}}^-) \) is the target Q-value, and \(\boldsymbol{\theta}_{\mathcal{A}}^-\) represents the target network’s parameters.

Similarly, the SOC analyst’s gradient update follows:
\begin{equation}
   \nabla_{\boldsymbol{\theta}_{\mathcal{A}}} L_{\mathcal{A}}(\boldsymbol{\theta}_{\mathcal{A}}) = \mathbb{E}_{s, a_{\mathcal{A}} \sim \rho(\cdot)} \left[ \left( y_{\mathcal{A}} - Q^0_{\mathcal{A}}(s, a_{\mathcal{A}}; \boldsymbol{\theta}_{\mathcal{A}}) \right) \nabla_{\boldsymbol{\theta}_{\mathcal{A}}} Q^0_{\mathcal{A}}(s, a_{\mathcal{A}}; \boldsymbol{\theta}_{\mathcal{A}}) \right].
\end{equation}
\begin{theorem}
As the number of attack graph nodes \( N \to \infty \), the Q-value function for the SOC analyst under CHT-DQN is lower bounded by the Q-value function under DQN, assuming a stationary and known attack strategy.
\end{theorem}

\begin{proof}
(See Appendix for detailed derivation.)
\end{proof}

This theorem demonstrates that as the attack graph complexity increases, CHT-DQN maintains a performance advantage by anticipating attacker strategies, making it a scalable and effective cybersecurity framework for cloud-based SOC operations.
\begin{figure}[t]
    \centering
    \begin{subfigure}{0.9\textwidth}
        \centering
        \includegraphics[width=\textwidth]{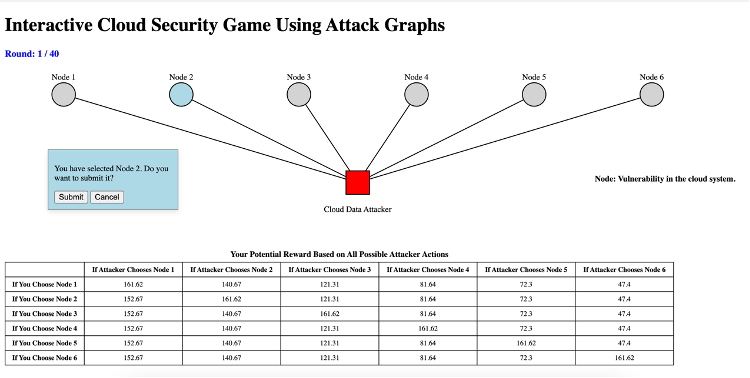}
        \caption{\textbf{Reward-Aware SOC Analyst $\mid$ DQN Attacker.} Participants defend cloud vulnerabilities based on immediate reward feedback without transition probabilities.}
        \label{fig:web-dqn}
    \end{subfigure}
    
    \vspace{5mm} 

    \begin{subfigure}{0.9\textwidth}
        \centering
        \includegraphics[width=\textwidth]{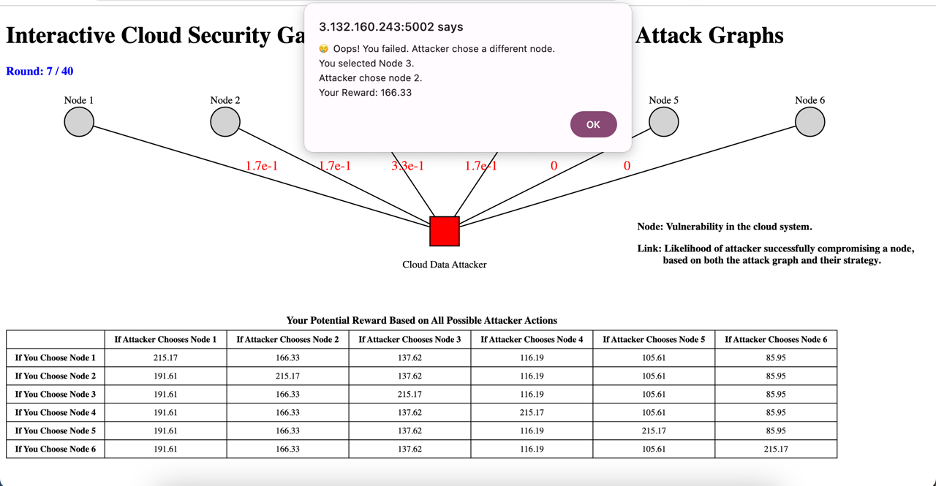}
        \caption{\textbf{Reward + Transition-Aware SOC Analyst $\mid$ DQN Attacker.} Participants receive both reward feedback and estimated attack transition probabilities, aligning with CHT-DQN.}
        \label{fig:web-cht}
    \end{subfigure}
    
    \caption{Human-Interactive Web-Based DRL Games on Amazon MTurk. The SOC analyst (defender) interacts with a DQN-based attacker under two different strategic information settings.}
    \label{fig:web-game}
\end{figure}

\section{Human-Interactive Web-Based DRL Game Scenarios}

\paragraph{Reward-Aware SOC Analyst $\mid$ DQN Attacker}
Participants act as SOC analysts, making real-time defensive decisions in a cloud security environment against a DQN-based attacker. They receive only reward feedback after selecting a defensive action, based on their chosen node and the attacker's selected target. This setup simulates SOC analysts refining their defensive strategies through experience (Figure~\ref{fig:web-dqn}).

\paragraph{Reward + Transition-Aware SOC Analyst $\mid$ DQN Attacker}
This setting provides both reward feedback and estimated attack transition probabilities, computed from Equation~\ref{eq:trans_prob}. This additional information enables SOC analysts to predict the attacker's behavior, improving cyber defense strategies through AI-assisted SOC decision-making (Figure~\ref{fig:web-cht}).

\paragraph{Participants and Task Procedure on Amazon MTurk}
This study received Institutional Review Board (IRB) approval, ensuring ethical compliance and minimal risk. All participants provided informed consent, and data collection adhered to cybersecurity research ethics. 

A total of 83 participants were recruited, with 80 completing all 40 rounds (40 per game variant). Three incomplete sessions were excluded from the analysis.

\paragraph{Qualification Requirements}
Participants met the following eligibility criteria:
\begin{itemize}
    \item HIT Approval Rate above 80\%.
    \item Country of residence: Australia, India, United Kingdom, or United States.
    \item Employment in Software \& IT Services to ensure relevant cybersecurity knowledge.
\end{itemize}

\paragraph{Experimental Phases}
The experiment was conducted in two phases:
\begin{itemize}
    \item \textbf{Phase 1:} Initial pilot with 10 participants per game for preliminary evaluation.
    \item \textbf{Phase 2:} Full-scale study with 30 participants per game. 
\end{itemize}
The “Reward + Transition-Aware” variant ran from 9:18 am to 10:13 am PDT, while the “Reward-Aware” variant ran from 9:18 am to 10:05 am PDT.

Participants were given instructions, played 40 rounds, and received a final score and confirmation code for compensation.

\paragraph{Data Collection and Metrics}
We recorded:
\begin{itemize}
    \item \textit{Participant ID:} Anonymized unique identifier.
    \item \textit{State and Action Histories:} SOC analyst and attacker decisions per round.
    \item \textit{Reward Histories:} Cumulative SOC analyst and attacker rewards.
    \item \textit{Attacker Policy and Exploitation Probabilities:} Estimations of level-0 attack behavior.
    \item \textit{Data Protection and Exploitation Costs:} Defensive resource allocations and attacker exploitation metrics.
    \item \textit{Performance Metrics:} Weighted data protection ratios, response times, and final scores.
\end{itemize}

This dataset enabled an in-depth analysis of real-time decision-making in both game settings, demonstrating the effectiveness of CHT-DQN for SOC-driven cloud security.

\begin{figure*}[t]
    \centering
    \subfloat[Convergence Analysis of SOC Analyst (Defender) over Time. The magenta line (CASE 1) shows the highest performance with CHT-DQN against a random attacker, followed by CHT-DQN against DQN (CASE 2). The lime and dark green lines (CASES 3 and 4) indicate lower performance for standard DQN SOC analysts.]
      {\includegraphics[width=0.45\textwidth]{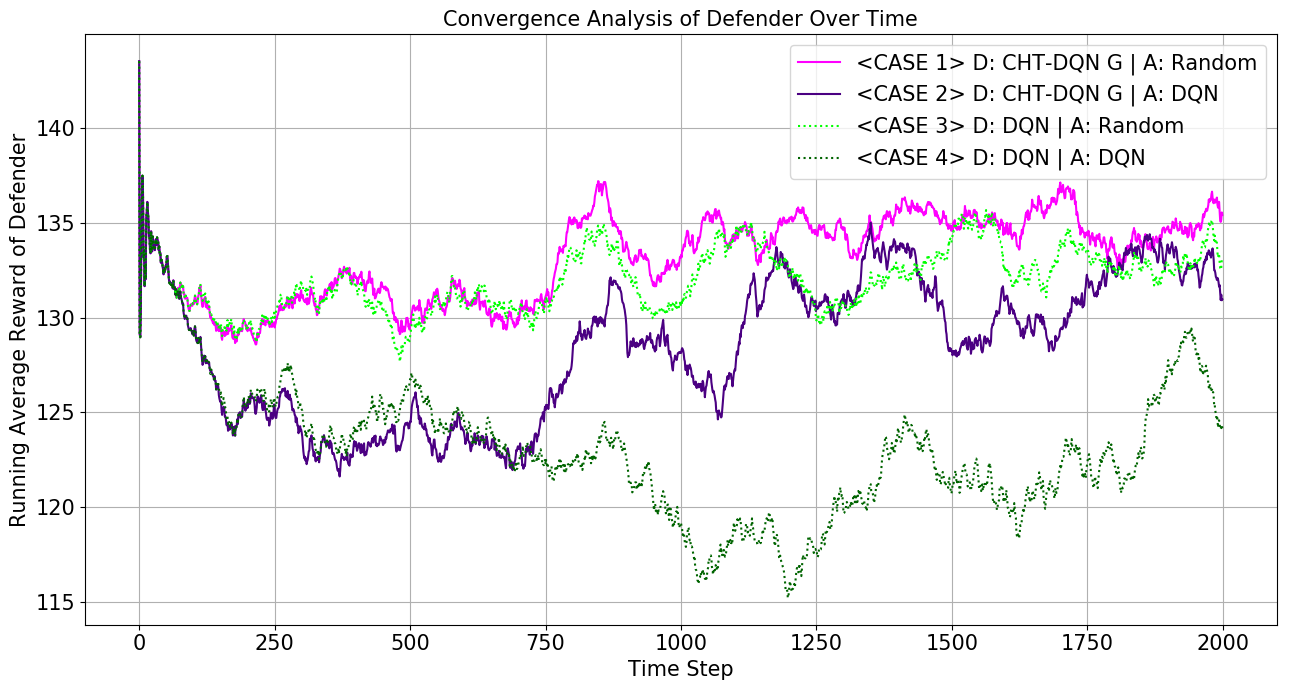}
        \label{fig:defender-convergence}}
    \hfill
    \subfloat[Convergence Analysis of Attacker over Time. The magenta line (CASE 1) shows the lowest attacker rewards against the CHT-DQN SOC analyst, while the indigo line (CASE 2) shows higher rewards against CHT-DQN. The lime and dark green lines (CASES 3 and 4) show the highest rewards against DQN SOC analysts.]{
        \includegraphics[width=0.45\textwidth]{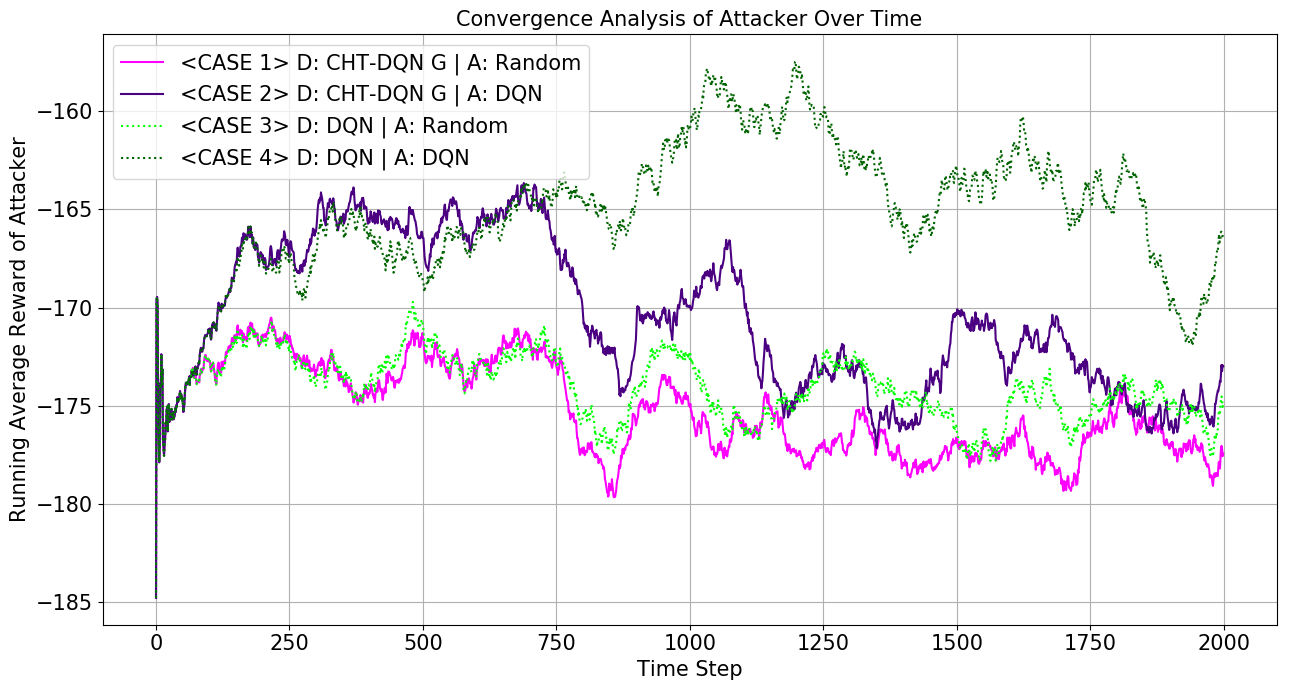}
        \label{fig:attacker-convergence}}
    \caption{Convergence Analysis: Running Average Rewards Over Time for SOC Analyst (Defender) and Attacker in Different Scenarios with a 6-Node AG. Timesteps 0 to 1000 represent training with epsilon-greedy, and timesteps 1000 to 2000 represent evaluation with pure exploitation. The results were averaged over 10 random seeds.}
\end{figure*}

\section{Experiments}

We evaluate the efficacy of Human-AI collaboration in SOCs using our proposed Cognitive Hierarchy Theory-driven Deep Q-Learning (CHT-DQN). Our study consists of two experimental settings:

\begin{itemize}
    \item \textbf{Deep Reinforcement Learning (DRL) Simulations:} Fully automated simulations where CHT-DQN-enhanced SOC analysts defend against AI-driven cyber threats in controlled environments.
    \item \textbf{Human-Interactive Web-Based DRL Experiments:} Large-scale human-in-the-loop experiments on Amazon MTurk, where participants act as SOC analysts defending cloud systems against an AI-driven attacker.
\end{itemize}

\subsection{Experimental Setup}

We implemented the SOC simulation and reinforcement learning framework in Python 3 using PyTorch. The framework integrates CHT-based decision-making to improve AI-augmented SOC defense strategies.

\paragraph{Neural Network Architecture.} 
The SOC analyst and attacker models use three hidden layers with 64, 128, and 256 neurons, respectively, and ReLU activations. The output layer maps cloud security states to optimal defense actions.

\paragraph{Training Parameters.}
The models were trained using the following hyperparameters:

\begin{itemize}
    \item \textbf{Batch size:} 64
    \item \textbf{Experience buffer size:} $1 \times 10^6$
    \item \textbf{Discount factor:} $\gamma = 0.98$ (prioritizing long-term security posturing)
    \item \textbf{Cognitive hierarchy depth:} $\beta = 1$
    \item \textbf{Learning rate:} $5 \times 10^{-2}$
\end{itemize}

The data size on attack graph nodes $\mathcal{G}$ ranged from 1 to 10, with a noise factor $\epsilon$ drawn from $[0,1]$. The defense and attack costs at node $i$ were defined as:
\begin{equation}
    b_i \pm \epsilon, \quad \hat{b}_i \pm \epsilon.
\end{equation}

We set $\alpha_{\mathcal{D},i} = 10$, $\alpha_{\mathcal{A},i} = 10$, $\beta_{\mathcal{D},i} = 1$, and $\beta_{\mathcal{A},i} = 1$, ensuring adaptive attack-defense interactions.

\begin{figure*}[t]
    \centering     
    \subfloat[Average Weighted Data Protection Ratio across different attack graph sizes. This figure shows how data protection improves as the number of nodes in the attack graph increases. CHT-DQN (Case 1: magenta line, Case 2: indigo line) consistently achieves higher data protection than standard DQN (Case 3: lime dotted line, Case 4: dark green dotted line). Error bars represent the standard deviation over 10 random seeds.]{
        \includegraphics[width=0.45\textwidth]{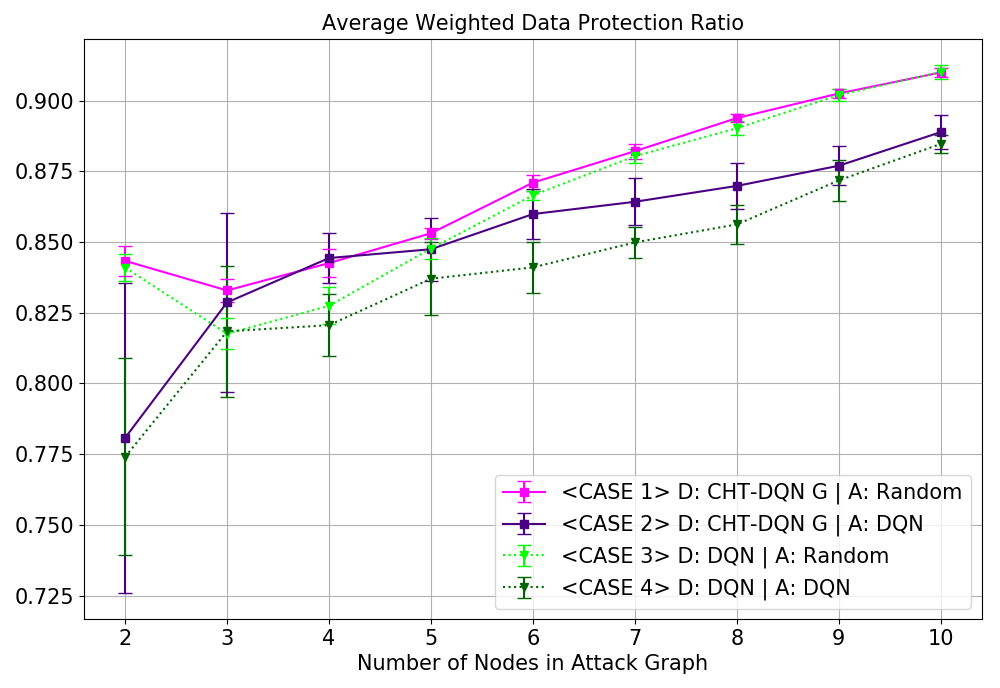}
    \label{fig:protect}}
    \hfill
    \subfloat[Normalized Discrepancy in Action Frequencies between SOC analyst and attacker as attack graph complexity increases. The discrepancy decreases with more complex attack graphs. Case 1 (magenta) shows the highest discrepancy, and Case 4 (dark green) the least for smaller AGs. For larger AGs, Case 3 (lime) has the least discrepancy, while Case 2 (indigo) shows the highest. Error bars represent standard deviation over 10 random seeds.]
        {\includegraphics[width=0.45\textwidth]{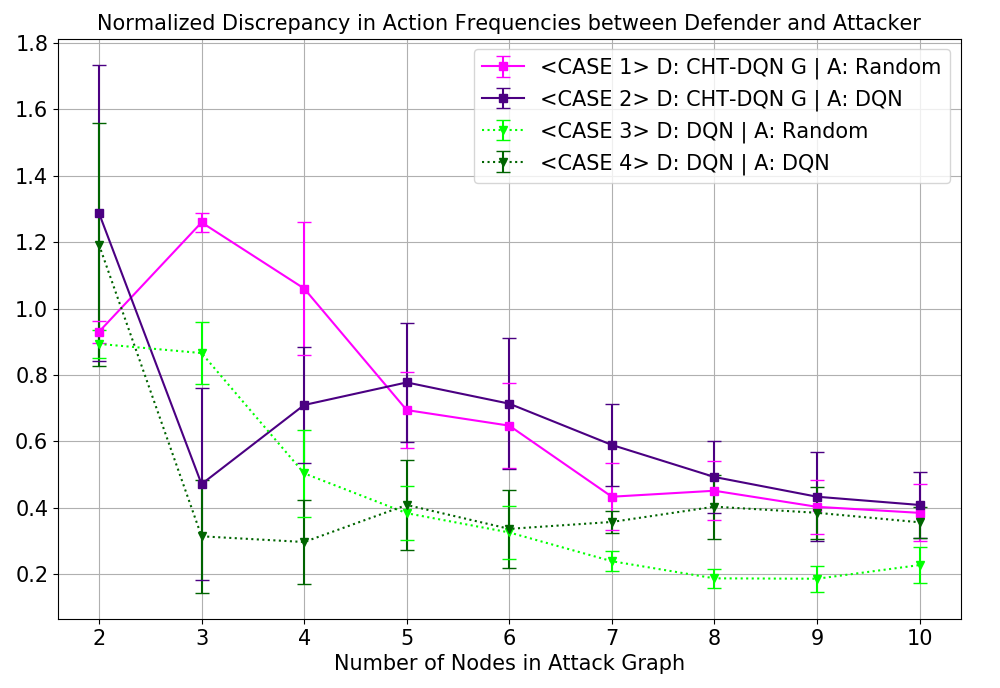}
        \label{fig:actions}}

    \caption{Performance Comparison of CHT-DQN and DQN SOC Analyst Strategies. The figures illustrate the advantages of using CHT-DQN over standard DQN in terms of data protection and action alignment.}\label{fig:cht-dqn}
\end{figure*}

\subsection{Deep Reinforcement Learning Scenario-Based Experiments}

We simulated interactions in a cloud-based SOC environment, where the AI-powered SOC analyst uses an attack graph (AG) to strategize against adversarial threats. The AG dynamically evolves, with nodes representing vulnerabilities and edges modeling exploit likelihoods.

\paragraph{Experimental Conditions.} 
We varied the number of attack graph nodes, $N \in [2,10]$, to analyze how defense strategies scale with complexity. Exploitation likelihoods $\mathbb{P}(v)$ in $\mathcal{G}$ were updated every 100 steps. 

To balance exploration and exploitation, we employed an $\epsilon$-greedy policy, initialized at $\epsilon=1$ and decaying to $\epsilon=0.05$ at a decay rate of 2. Each SOC simulation ran for 2000 time steps, averaged over 10 random seeds.

\paragraph{Evaluated Defense Strategies.}
We compared four SOC defense strategies:
\begin{itemize}
    \item \textbf{Case 1:} CHT-DQN SOC analyst vs. Random attacker
    \item \textbf{Case 2:} CHT-DQN SOC analyst vs. DQN attacker
    \item \textbf{Case 3:} Standard DQN SOC analyst vs. Random attacker
    \item \textbf{Case 4:} Standard DQN SOC analyst vs. DQN attacker
\end{itemize}

These settings model real-world SOC operations, enabling comparisons between adaptive SOC analysts (CHT-DQN) and traditional reinforcement learning (DQN) under varying adversarial conditions.

\paragraph{Findings.}
Figure \(\ref{fig:defender-convergence}\) shows the running average reward of the SOC analyst over time, while Figure \(\ref{fig:attacker-convergence}\) presents the running average reward of the attacker. These figures provide insights into strategy evolution and convergence during training. The running average is computed over the last 100 timesteps, with the interval from timestep 0 to 1000 representing the training phase (epsilon-greedy exploration-exploitation) and from timestep 1000 to 2000 representing the evaluation phase (exploitation-based decision-making). An AG with 6 nodes was used for these experiments.

In the SOC analyst's convergence analysis (Figure \(\ref{fig:defender-convergence}\)), Case 1 (magenta line) shows the highest rewards for the SOC analyst using CHT-DQN against a random attacker, indicating effective learning and anticipation of random attacks. Case 2 (indigo line) shows slightly lower rewards for CHT-DQN against a DQN attacker, reflecting a more complex learning process. Case 3 (lime dotted line) shows worse performance for standard DQN against a random attacker, highlighting the advantage of CHT. Case 4 (dark green dotted line) shows the lowest performance for both SOC analyst and attacker using standard DQN, suggesting they struggle without CHT. In the attacker's convergence analysis (Figure \(\ref{fig:attacker-convergence}\)), Case 1 shows lower rewards for a random attacker against a CHT-DQN SOC analyst, indicating effective defense. Case 2 shows higher rewards for a DQN attacker against CHT-DQN but still lower than against standard DQN. Case 3 shows higher rewards for a random attacker against standard DQN, indicating less effective defense. Case 4 shows the highest rewards for a DQN attacker against a standard DQN SOC analyst, demonstrating the latter's struggle. Overall, these results confirm that CHT-DQN improves real-time SOC decision-making, leading to more effective human-AI collaboration in mitigating cyber threats.

In Figure \(\ref{fig:protect}\), the average weighted data protection ratio across the cloud system is shown, formulated as \(\sum_i (\delta^t_i \cdot b_i)/\sum b_i\). This metric represents the proportion of cloud data protected at each timestep, averaged over the simulation duration, highlighting how AI-driven SOC analysts safeguard cloud infrastructure against APTs. As the number of nodes in the AG increases, the trend shows improved data protection, indicating that higher complexity in the cloud system allows the algorithms to better safeguard data by reducing the impact of any single node being exploited.
Remarkably, scenarios, where the SOC analyst utilized CHT-DQN (Case 1: magenta line, and Case 2: indigo line), outperformed those using DQN (Case 3: lime dotted line, and Case 4: dark green dotted line). This confirms that CHT-DQN model enhances SOC decision-making. The highest data protection is achieved when the attacker takes random actions against a SOC analyst armed with CHT-DQN (Case 1), underscoring the value of cognitive hierarchy in defense strategies. Conversely, the lowest data protection occurs when both the SOC analyst and attacker employ DQN at the same cognitive level (Case 4). By incorporating CHT, the SOC analyst's model of the attacker is enhanced by adjusting the attacker's Q-values based on the frequency of previously successful defensive actions, refining the SOC analyst's policy by best-responding to a more accurate model of the attacker.

Figure \(\ref{fig:actions}\) illustrates the normalized discrepancy in action frequencies between the SOC analyst and attacker as the number of nodes in \(\mathcal{G}\) varies. Generally, the discrepancy decreases as \(\mathcal{G}\) becomes more complex, suggesting that with more nodes, the SOC analyst's actions align more closely with the attacker's. For smaller AGs, Case 4 shows the least frequency of action discrepancy, indicating predictable interactions when both parties use similar strategic algorithms. However, Case 1 exhibits the greatest discrepancy, indicating that a sophisticated defense strategy like CHT-DQN may not align well with an unpredictable, random attacker. As AG complexity grows, Case 3 shows the least discrepancy, implying that a standard defense approach may align better against random attacks in more complex environments. Case 2 shows the highest discrepancy, suggesting that CHT-DQN's complexity could be less effective in anticipating a DQN-based attack in a more complex environment.

In both figures, error bars represent the standard deviation over 10 random seeds, providing insight into result variability and confidence. In Figure \(\ref{fig:protect}\), the increasing trend with a relatively lower standard deviation as the number of nodes increases indicates stable and consistent performance improvements with more complex attack graphs. Conversely, in Figure \(\ref{fig:actions}\), higher standard deviation in scenarios with fewer nodes suggests more variability in performance, which stabilizes as the number of nodes increases, indicating the convergence in strategic actions between the SOC analyst and attacker in more complex environments. This pattern underscores the necessity of considering both mean performance and variability to fully understand the efficacy and reliability of defensive strategies in various scenarios.
\begin{figure}[th]
\centering
\includegraphics[width=0.6\textwidth]{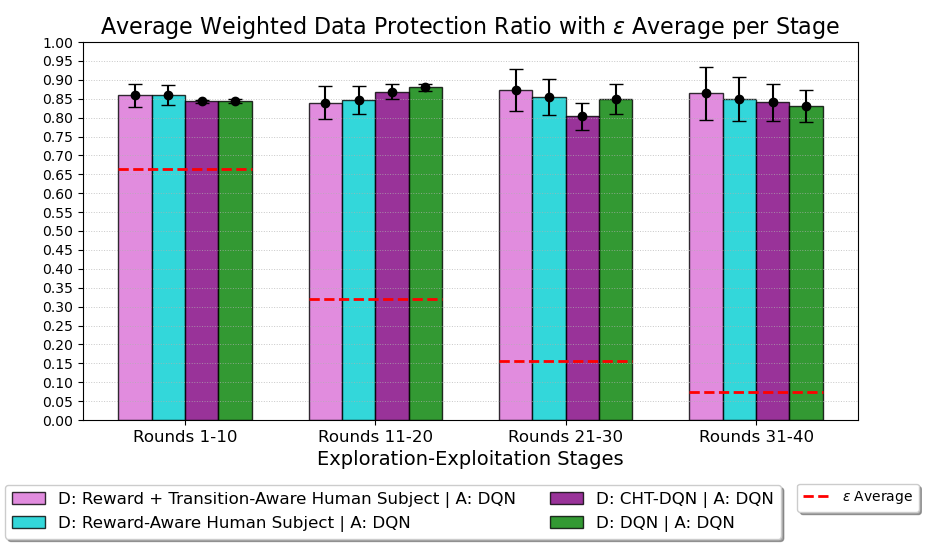}
\caption{Average Weighted Data Protection Ratio with $\epsilon$ Average per Stage across 40 Rounds. This figure compares four defense scenarios, highlighting the advantages of human-AI collaboration scenarios 1 (orchid) and 2 (dark turquoise) over fully automated AI-driven decision-making scenarios 3 (purple) and 4 (green). The red dashed line indicates the average $\epsilon$ values across four exploration-exploitation stages, highlighting the shift from exploration to exploitation over time.}
\label{fig:mturk-protect}
\end{figure}
\begin{figure}[tb]
\centering
\includegraphics[width=0.8\textwidth]{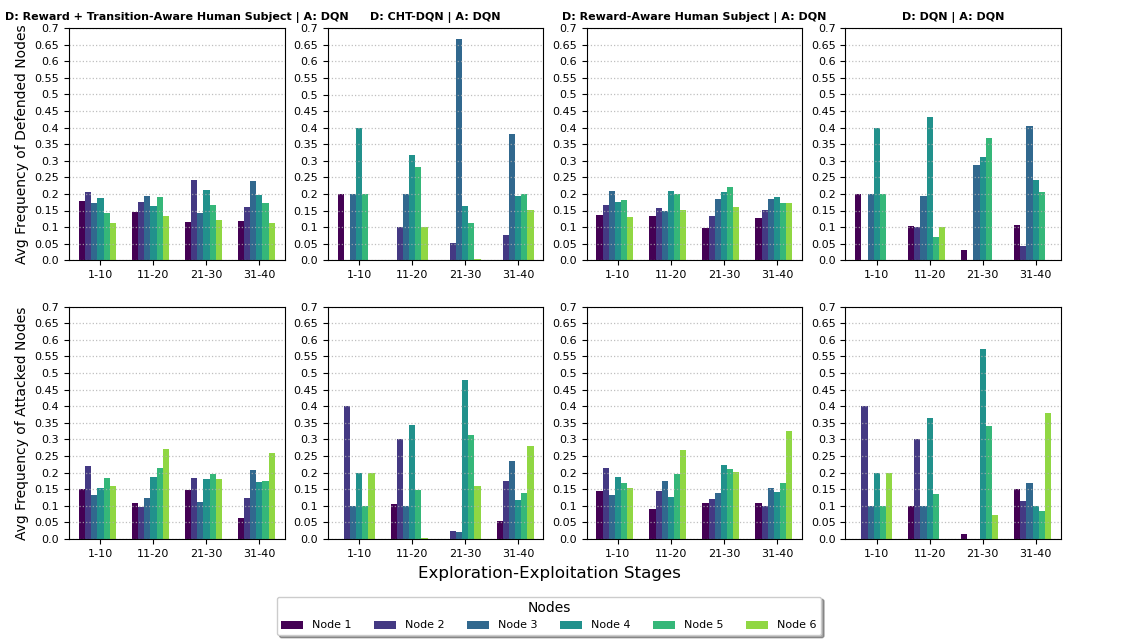}
\caption{Average Frequency of Defended and Attacked Nodes across Different Exploration-Exploitation Stages. The scenarios are as follows: "D: Reward + Transition-Aware Human Subject $\big|$ A: DQN," "D: Reward-Aware Human Subject $\big|$ A: DQN," "D: CHT-DQN $\big|$ A: DQN," and "D: DQN $\big|$ A: DQN." The top row shows the average frequency of defended nodes, while the bottom row shows the frequency of attacked nodes, with colors indicating specific nodes in the attack graph. Human-in-the-loop SOC analysts (Scenarios 1 and 2) demonstrate more dynamic adaptation compared to AI-driven models (Scenarios 3 and 4).}
\label{fig:mturk-actions}
\end{figure}

\subsection{Human-Interactive Web-Based DRL Experiments on Amazon MTurk}

To validate our human-AI collaboration framework in SOCs, we conducted a large-scale human-in-the-loop experiment on Amazon Mechanical Turk (MTurk). The two interactive web-based games were implemented on an AWS t3.medium EC2 instance, which provided the necessary scalability, low latency, and reliable performance for real-time human-AI interactions. Participants accessed the games via the links \url{http://3.132.160.243:5001/} and \url{http://3.132.160.243:5002/} on MTurk. In this human-AI SOC simulation, participants played the role of SOC analysts in a dynamic cloud defense scenario designed to counter AI-driven attackers. The goal was to examine how human SOC analysts behave when confronted with DQN attackers, using two different scenarios modeled by DQN and CHT-DQN algorithms.

Participants were tasked with defending cloud nodes against APTs over 40 rounds per session. The security environment was modeled using a dynamically constructed 6-node AG, representing cloud infrastructure vulnerabilities. The attacker's strategy was AI-driven, adapting its tactics using a DQN model, simulating a realistic adaptive adversary in SOC operations. 

\paragraph{Evaluated Scenarios.}
We evaluated four SOC analyst conditions:

\begin{itemize}
    \item \textbf{Scenario 1:} Reward + Transition-Aware Human SOC Analyst vs. DQN attacker
    \item \textbf{Scenario 2:} Reward-Aware Human SOC Analyst vs. DQN attacker
    \item \textbf{Scenario 3:} Fully automated CHT-DQN SOC analyst vs. DQN attacker
    \item \textbf{Scenario 4:} Fully automated DQN SOC analyst vs. DQN attacker
\end{itemize}
In Scenarios 1 and 2, human SOC analysts engaged in real-time cybersecurity decision-making, facing automated DQN attackers across 40 rounds. These scenarios evaluated how human-AI collaborative decision-making compares to fully automated SOC analysts.
For Scenarios 3 and 4, results were averaged over 40 random seeds to ensure statistical robustness in AI-driven SOC performance evaluations.

In all four scenarios, the automated SOC analyst or attacker ran their algorithm with $\epsilon$-greedy policy with $\epsilon$ starting at $0.9$ to $0.05$, indicating starting more exploration and ending in more exploitation.
The web interface allowed us to track SOC analyst actions, failure events, and subsequent responses. The study yielded several key insights, illustrated through four figures derived from both human SOC analysts and AI-driven simulation data. 

\paragraph{Findings.}
Figure \ref{fig:mturk-protect} presents the average weighted data protection ratio over four exploration-exploitation stages, comparing four distinct scenarios.  The attack graph employed consists of $6$ nodes, with the $\epsilon$-greedy parameter of the automated agents, shown as the red dashed line, indicating the transition from exploration to exploitation over time, especially in the later stages where exploitation becomes predominant.

In the Reward-Aware Human SOC analyst scenario, participants selected nodes to defend each round based on observing potential rewards without access to transition probabilities. In contrast, the Reward + Transition-Aware Human SOC analyst scenario provided participants with both potential rewards and transition probabilities derived from the CHT-DQN model. This additional information allowed SOC analysts at cognitive level-$1$ to anticipate attacker strategies at level-$0$ more effectively, particularly in later stages where exploitation is emphasized.

The results reveal a marked performance advantage in "D: Reward + Transition-Aware Human Subject $\big|$ A: DQN" (Scenario 1) over "D: CHT-DQN $\big|$ A: DQN" (Scenario 3) in the exploitation stage. This outcome suggests that access to CHT-DQN-driven transition probabilities—incorporating the attacker’s level-$0$ policy and attack graph likelihoods—enhances the SOC analyst’s ability to safeguard cloud data. The "D: Reward-Aware Human Subject $\big|$ A: DQN" (Scenario 2), similar to the fully simulated "D: DQN $\big|$ A: DQN" case (Scenario 4), shows a relatively lower average weighted data protection, indicating that potential rewards alone are insufficient for optimal defense against a strategic attacker model. This underscores the potential for interactive human-AI systems to apply CHT, effectively bridging human decision-making and advanced automated cloud defense strategies.

Figure \ref{fig:mturk-actions} compares the frequency of defended and attacked nodes across four scenarios over 40 rounds, showing how SOC analysts with different levels of strategic information adapt to protect a 6-node cloud attack graph. The top row displays "Average Frequency of Defended Nodes" and the bottom row "Average Frequency of Attacked Nodes" across four stages, each reflecting a transition from exploration to exploitation.

The "D: Reward + Transition-Aware Human Subject $\big|$ A: DQN" (Scenario 1) demonstrates the highest alignment with attacker actions, as the SOC analyst, informed by both potential rewards and transition probabilities from CHT-DQN, adapts strategically to anticipated attacks. This suggests a significant advantage in combining human intuition with CHT-DQN-driven insights for effective defense. In contrast, the CHT-DQN strategic advantage of Scenario 3, although automated, shows moderately effective alignment, as it incorporates probabilistic strategies but lacks the nuanced adaptability of human intuition.

The "D: Reward-Aware Human Subject $\big|$ A: DQN" (Scenario 2), where the SOC analyst sees only potential rewards, exhibits moderate alignment, indicating that while reward information alone helps, it lacks the strategic depth provided by transition probabilities. Scenario 4, which is fully automated and based solely on immediate rewards, displays the least alignment with attacker actions, highlighting the limitations of a purely reactive defense strategy.

Additionally, the figure shows that SOC analysts, in both scenarios 1 and 2, exhibit a more exploratory behavior than simulated SOC analysts, as indicated by the more balanced action frequencies across nodes. In contrast, some actions in the automated scenarios 3 and 4 have lower frequencies, suggesting a more rigid, exploitative approach. This pattern underscores that SOC analysts explore a broader range of defensive actions, potentially enabling more adaptive responses to dynamic attack strategies.

\begin{figure*}[t]
    \centering     
    \subfloat[Node reselection likelihood after failure for human and AI-based SOC analysts. Human participants show a lower likelihood of reselection after failure, consistent with the loss aversion and underweighting of potential future gains following a loss, as predicted by PT and CPT.]{\includegraphics[width=0.45\textwidth]{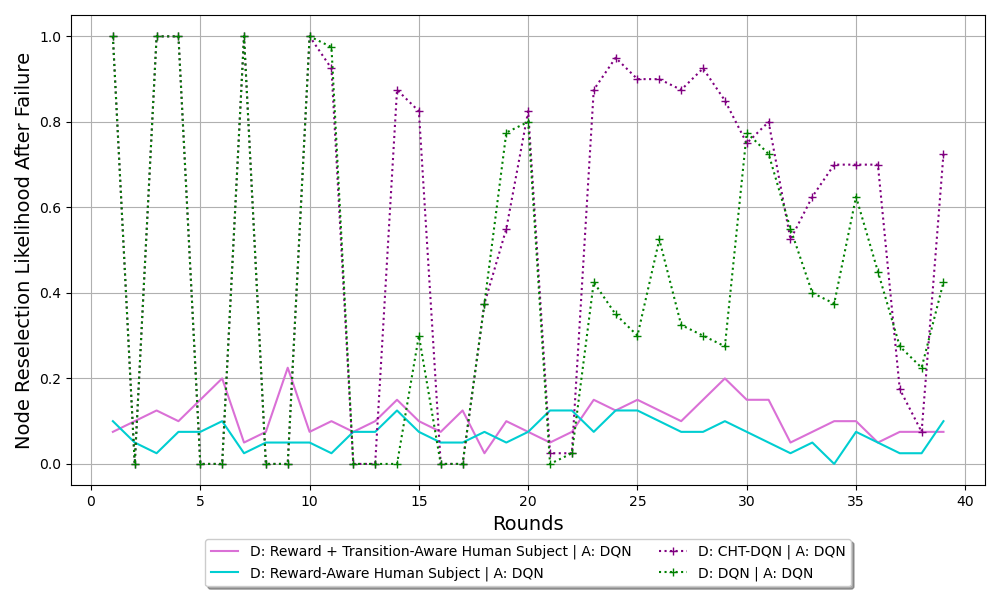}
    \label{fig:mturk-risk-aversion}}
    \hfill
    \subfloat[Node reselection likelihood after success for human and AI-based SOC analysts. Human participants show a continued tendency to reselect previously successful nodes, consistent with an overestimation of the likelihood of continued success in PT.]{\includegraphics[width=0.45\textwidth]{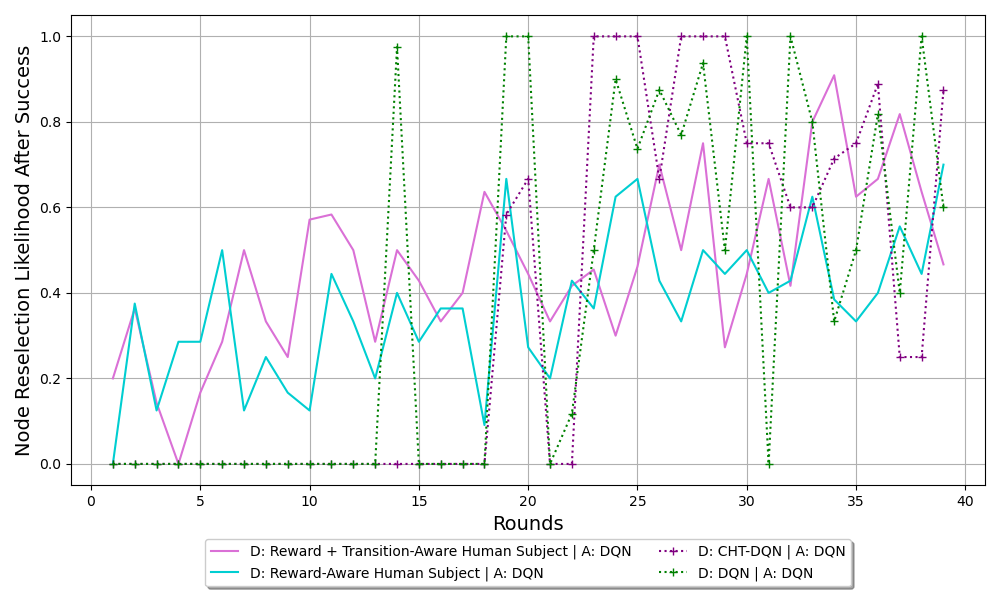}
        \label{fig:mturk-risk-seeking}}
    \caption{Node reselection likelihood by human and AI-based SOC analysts. Human behavior shows reduced reselection after failure and persistent reselection after success. This asymmetry is consistent with the framing effect and probability weighting in Prospect Theory (PT) and Cumulative Prospect Theory (CPT).}
\label{fig:mturk-pt}
\end{figure*}

Figure~\ref{fig:mturk-risk-aversion} shows that human SOC analysts in Scenarios 1 and 2 are less likely to reselect nodes that failed to protect data in the previous round, while Figure~\ref{fig:mturk-risk-seeking} reveals a continued tendency to reselect previously successful nodes. This asymmetry is well explained by Prospect Theory (PT) and Cumulative Prospect Theory (CPT) (Kahneman and Tversky \cite{kahneman1979prospect}; Tversky and Kahneman \cite{tversky1992advances}).

According to PT, the subjective value function is concave for gains and convex for losses, with losses perceived more strongly due to loss aversion. This amplifies the psychological cost of failure, discouraging immediate reselection. CPT further refines this perspective through nonlinear probability weighting, where the probability of success after a failure tends to be underweighted, whereas the likelihood of continued success is often overestimated. Together, these mechanisms explain the observed human behavior: avoidance of previously failed nodes and a persistent return to those that previously yielded success.

Table \ref{tab:elapsed-time} presents the average elapsed time for each scenario, revealing significant differences in time investment between human-interactive and fully automated scenarios. The scenarios with human SOC analysts (Scenario 1 and 2) recorded substantially higher average elapsed times, 377.93 and 459.89 seconds, respectively. This reflects the additional cognitive processing required for human participants to make strategic decisions, especially in scenarios where both reward information and transition probabilities are provided. In contrast, the fully automated scenarios (Scenario 3 and 4) completed in negligible time (0.09 and 0.08 seconds, respectively), as AI-based models make decisions instantaneously without the need for human-like deliberation. This disparity underscores the trade-off between strategic depth and response speed in human versus AI SOC analysts, with human-involved scenarios offering potentially more nuanced decision-making at the cost of increased time.
\begin{table}[t]
  \centering
  \caption{Average Elapsed Time Across Scenarios}
  \label{tab:elapsed-time}
  \begin{tabular}{lc}
    \toprule
    \textbf{Scenario} & \textbf{Elapsed Time (seconds)} \\
    \midrule
    Reward + Transition-Aware Human SOC Analyst & 377.93 \\
    Reward-Aware Human SOC Analyst & 459.89 \\
    CHT-DQN SOC Analyst & 0.09 \\
    DQN SOC Analyst & 0.08 \\
    \bottomrule
  \end{tabular}
\end{table}
The MTurk study provided crucial validation for the effectiveness of CHT-DQN models in cloud defense. By comparing human behavior against the simulated models, we observed that SOC analysts using CHT-DQN exhibited behavior closely mimicking human risk aversion and strategic adaptation patterns, confirming the model's applicability in real-world scenarios. These findings underscore the importance of incorporating cognitive models in AI-driven cybersecurity, as they better align with human decision-making processes and improve defense outcomes. These findings underscore the importance of human intuition and adaptability in strategic defense, as well as the distinct ways humans balance risk and exploitation compared to AI models.
\section{Conclusions}

This study introduces a Cognitive Hierarchy Theory-driven Deep Reinforcement Learning framework for human-AI collaboration in SOCs, demonstrating its effectiveness in mitigating AI-driven APT threats through adaptive decision-making. The integration of Attack Graphs with CHT-DQN enables SOC analysts to model adversarial behavior at multiple levels of reasoning, significantly outperforming standard DQN across both automated and human-in-the-loop defense scenarios. 
Key findings from human-AI interaction experiments on Amazon MTurk show that SOC analysts at cognitive level-1 with access to CHT-DQN-driven transition probabilities—incorporating the attacker’s level-0 policy and attack graph likelihoods—achieve higher average weighted data protection ratios and show better alignment with adaptive attackers, as evidenced by reduced action discrepancies. These experiments also align with principles from Prospect Theory and Cumulative Prospect Theory, as SOC analysts tend to avoid previously failed actions but persist with those that succeeded—highlighting cognitively grounded patterns in human-AI cybersecurity collaboration.
Simulation experiments further validate that CHT-DQN SOC analysts adapt more effectively to increasing attack graph complexity, demonstrating superior data protection capabilities against both random and strategic attackers. Increasing the AG complexity leads to higher weighted data protection ratios and lower action frequency discrepancies, reinforcing the ability of CHT-DQN to model and counter evolving adversarial strategies.

Convergence analysis indicates that CHT-DQN systematically outperforms standard DQN in mid-sized AGs, achieving higher cumulative rewards and more effective adaptation against dynamic attack strategies. Additionally, response time analysis highlights a critical trade-off between human-driven and AI-driven SOC decisions, with human analysts requiring more deliberation time but demonstrating superior strategic foresight, compared to the faster yet less adaptive AI-based models.

Furthermore, theoretical validation confirms that as AG complexity grows, the Q-values of CHT-DQN-based SOC analysts consistently surpass those of standard DQN, establishing CHT-DQN as a more effective framework for complex cybersecurity environments. These insights emphasize the critical role of cognitive modeling in deep reinforcement learning, providing a robust foundation for integrating real-time AI-driven decision support systems in SOCs.

By bridging cognitive reasoning and deep reinforcement learning, this work contributes a novel approach to AI-augmented cybersecurity, illustrating how human-AI collaboration can enhance real-time SOC decision-making. These findings pave the way for next-generation SOCs, where adaptive AI systems and human expertise jointly strengthen cloud security defenses against evolving cyber threats.

\section*{Acknowledgments}
The authors acknowledge the support of the Rutgers University Institutional Review Board (IRB) for approving the human-subject research conducted as part of this study. The research involving participants on Amazon Mechanical Turk (MTurk) was reviewed and approved under IRB Protocol No. Pro2024000556, determined to be minimal risk and granted expedited review status. All participants provided informed consent before participation.



\bibliographystyle{IEEEtran}
%

\bibliography{IEEEabrv,main}


\clearpage
\appendix

\section{Appendix}
\label{sec:appendix}

\paragraph{Theorem 4.1}
\textit{As the number of nodes \( N \to \infty \) in the attack graph \( \mathcal{G} \), the \( Q \)-value function for the SOC analyst under the CHT-DQN policy is lower bounded by the \( Q \)-value function for the SOC analyst under DQN, assuming the attack strategy is stationary and known.}

\begin{proof}
Let \( \mathcal{G} \) be an attack graph with \( N \) nodes. The proof follows the steps below:

\textbf{1. Boundedness of Utilities and Discounted Rewards:}
   We establish that the utility function and the expected cumulative reward remain bounded for both CHT-DQN and standard DQN, ensuring comparability as \( N \to \infty \).

\textbf{2. Inductive Argument:}
   Assume at step \( t \), the expected utility of the SOC analyst under CHT-DQN is greater than or equal to that under standard DQN:
   \[
   \mathbb{E}[R_{\mathcal{D}}(s) | \pi_{\mathcal{D}}^{\text{CHT-DQN}}] \geq \mathbb{E}[R_{\mathcal{D}}(s) | \pi_{\mathcal{D}}^{\text{DQN}}].
   \]
   By induction, this holds for all future steps.

\textbf{3. Value Function Comparison:} 
   Since the expected cumulative reward is positively correlated with the value function, it follows that:
   \[
   V^{\text{CHT-DQN}}_{\mathcal{D}}(s) \geq V^{\text{DQN}}_{\mathcal{D}}(s).
   \]

\textbf{4. Q-Value Function Relation:}
   Given the Bellman equation and the relationship between the value function and the Q-function, we conclude:
   \[
   Q^{\text{CHT-DQN}}_{\mathcal{D}}(s, a_{\mathcal{D}}) \geq Q^{\text{DQN}}_{\mathcal{D}}(s, a_{\mathcal{D}}).
   \]

Thus, the CHT-DQN policy provides a stronger lower bound on the SOC analyst’s $Q$-values compared to standard DQN, completing the proof.

\subsection{Boundedness of Utilities and Cumulative Discounted Reward}
We first establish that the utility functions and cumulative rewards are bounded. The data sizes \( \hat{b}_i \), \( b_i \), and costs \( c_{\mathcal{A},i}, c_{\mathcal{D},i} \) are finite and bounded for all \( i \). The weighting coefficients \( \alpha_{\mathcal{A},i}, \beta_{\mathcal{A},i}, \alpha_{\mathcal{D},i}, \beta_{\mathcal{D},i} \) are also finite. Additionally, the security status indicator \( \delta^t_i \) is binary (\( 0 \) or \( 1 \)). Therefore, the attacker’s and SOC analyst’s utilities, \( u^t_{\mathcal{A}} \) (or \( R^t_{\mathcal{A}} \)) and \( u^t_{\mathcal{D}} \) (or \( R^t_{\mathcal{D}} \)), remain bounded for any \( t \) and any number of nodes \( N \).

The cumulative discounted reward for a policy \( \pi_{\mathcal{D}} \) starting from state \( s \) is:
\begin{equation}
    V^\pi_{\mathcal{D}}(s) = \mathbb{E} \left[ \sum_{t=0}^{\infty} \gamma^t R^t_{\mathcal{D}}(s^t, a^t_{\mathcal{D}}, s^{t+1}) \mid s^0 = s \right],
\end{equation}
where \( \gamma \) is the discount factor (\( 0 \leq \gamma < 1 \)). Since \( R_{\mathcal{D}} \) consists of bounded utilities \( u^t_{\mathcal{A}} \) and \( u^t_{\mathcal{D}} \), the series \( \sum_{t=0}^{\infty} \gamma^t R^t_{\mathcal{D}}(s^t, a^t_{\mathcal{D}}, s^{t+1}) \) is a discounted sum of bounded terms. A discounted sum converges when \( \gamma < 1 \), ensuring that \( V^\pi_{\mathcal{D}}(s) \) remains bounded for any policy \( \pi_{\mathcal{D}} \) and initial state \( s \). Consequently, as \( N \to \infty \), the expected cumulative reward remains finite, allowing for a direct comparison of CHT-DQN and DQN performance.

\subsection{Proof by Induction: CHT-DQN Lower Bound on DQN}
We proceed by induction to show that the expected utility of the SOC analyst under the CHT-DQN policy is lower bounded by that under the DQN policy.

The expected reward (or utility) of the SOC analyst under CHT-DQN and DQN policies are defined as:

\begin{equation}
    \mathbb{E}[R_{\mathcal{D}}^{\text{CHT-DQN}}] = \sum_{a_{\mathcal{D}}, a_{\mathcal{A}}, s'} \mathbb{P}(s', a_{\mathcal{A}} | s, a_{\mathcal{D}}) \cdot R_{\mathcal{D}}(s, a_{\mathcal{D}}, a_{\mathcal{A}}, s'),
\end{equation}

\begin{equation}
    \mathbb{E}[R_{\mathcal{D}}^{\text{DQN}}] = \sum_{a_{\mathcal{D}}, s'} \mathbb{P}(s' | s, a_{\mathcal{D}}) \cdot R_{\mathcal{D}}(s, a_{\mathcal{D}}, s').
\end{equation}

For the CHT-DQN case, the probability distribution over the attacker's actions \( a_{\mathcal{A}} \) is conditioned on the SOC analyst’s policy, whereas for standard DQN, the attacker's strategy is treated as independent. Since the CHT-DQN explicitly integrates attacker strategy modeling, it accounts for adversarial responses more effectively than a purely reactive DQN.

\textbf{Base Case:} For a small attack graph with \( N = 1 \), both policies make decisions based on immediate expected rewards. Since both rely on similar reinforcement learning updates, the difference in expected reward is minimal.

\textbf{Inductive Step:} Assume that for some \( N = k \), the expected utility satisfies:
\begin{equation}
    \mathbb{E}[R_{\mathcal{D}}^{\text{CHT-DQN}}] \geq \mathbb{E}[R_{\mathcal{D}}^{\text{DQN}}].
\end{equation}
For \( N = k+1 \), the transition probabilities \( \mathbb{P}(s' | s, a_{\mathcal{D}}) \) in DQN remain fixed, whereas in CHT-DQN, they dynamically adjust to anticipated attacker behaviors. By incorporating a probabilistic model of adversary decision-making, CHT-DQN optimizes long-term expected rewards better than the standard DQN. Hence, the bound is preserved, proving that:
\begin{equation}
    \forall N, \quad \mathbb{E}[R_{\mathcal{D}}^{\text{CHT-DQN}}] \geq \mathbb{E}[R_{\mathcal{D}}^{\text{DQN}}].
\end{equation}

Thus, as \( N \to \infty \), the CHT-DQN policy ensures that the SOC analyst's performance remains at least as effective as, if not superior to, a model-free DQN approach.
\subsubsection{Base Case: \( N = 1 \)}
For an attack graph \( \mathcal{G} \) with a single node \( N = 1 \), the policies of CHT-DQN and standard DQN are equivalent, meaning the expected utilities satisfy:
\begin{equation}
    \mathbb{E}[R_{\mathcal{D}}^{\text{CHT-DQN},1}] = \mathbb{E}[R_{\mathcal{D}}^{\text{DQN},1}].
\end{equation}

\subsubsection{Inductive Hypothesis}
Assume that for \( N = k \), the expected reward for the SOC analyst under CHT-DQN is greater than that under DQN:
\begin{equation}
    \mathbb{E}[R_{\mathcal{D}}^{\text{CHT-DQN},k}] > \mathbb{E}[R_{\mathcal{D}}^{\text{DQN},k}].
\end{equation}
Expanding this expectation, we obtain:
\begin{equation}
    \sum_{a_{\mathcal{D}}, s'} \sum_{a_{\mathcal{A}}=1}^{k} \mathbb{P}(s', a_{\mathcal{A}} | s, a_{\mathcal{D}}) R_{\mathcal{D}}(s, a_{\mathcal{D}}, a_{\mathcal{A}}, s') 
    > \sum_{a_{\mathcal{D}}, s'} \mathbb{P}(s' | s, a_{\mathcal{D}}) R_{\mathcal{D}}(s, a_{\mathcal{D}}, s').
\end{equation}

\subsubsection{Inductive Step: \( N = k+1 \)}
For \( N = k+1 \), we aim to show that:
\begin{equation}
    \mathbb{E}[R_{\mathcal{D}}^{\text{CHT-DQN},k+1}] > \mathbb{E}[R_{\mathcal{D}}^{\text{DQN},k+1}].
\end{equation}
Expanding the expectation:
\begin{equation}
    \mathbb{E}[R_{\mathcal{D}}^{\text{CHT-DQN},k+1}] =
    \sum_{a_{\mathcal{D}}, s'} \sum_{a_{\mathcal{A}}=1}^{k+1} \mathbb{P}(s', a_{\mathcal{A}} | s, a_{\mathcal{D}}) R_{\mathcal{D}}(s, a_{\mathcal{D}}, a_{\mathcal{A}}, s').
\end{equation}
Splitting the summation:
\begin{equation}
    \mathbb{E}[R_{\mathcal{D}}^{\text{CHT-DQN},k+1}] =
    \sum_{a_{\mathcal{D}}, s'} \left[ 
    \sum_{a_{\mathcal{A}}=1}^{k} \mathbb{P}(s', a_{\mathcal{A}} | s, a_{\mathcal{D}}) R_{\mathcal{D}}(s, a_{\mathcal{D}}, a_{\mathcal{A}}, s')
    + \mathbb{P}(s', a_{\mathcal{A}} = k+1 | s, a_{\mathcal{D}}) R_{\mathcal{D}}(s, a_{\mathcal{D}}, a_{\mathcal{A}}=k+1, s')
    \right].
\end{equation}
Since the additional node’s impact is limited as \( N \) increases, the relative gain from including \( N = k+1 \) remains non-negative:
\begin{equation}
    \mathbb{E}[R_{\mathcal{D}}^{\text{CHT-DQN},k+1}] - \mathbb{E}[R_{\mathcal{D}}^{\text{CHT-DQN},k}]
    > \mathbb{E}[R_{\mathcal{D}}^{\text{DQN},k+1}] - \mathbb{E}[R_{\mathcal{D}}^{\text{DQN},k}].
\end{equation}
Thus, by induction, CHT-DQN outperforms DQN for all \( N \).

\subsection{Correlation Between Expected Cumulative Reward and Value Function}
Since \(\mathbb{E}[R_{\mathcal{D}}(s) | \pi_{\mathcal{D}}^{\text{CHT-DQN}}] > \mathbb{E}[R_{\mathcal{D}}(s) | \pi_{\mathcal{D}}^{\text{DQN}}]\) holds for all \( s \), we extend this to the value function:
\begin{equation}
    V^{\text{CHT-DQN}}_{\mathcal{D}}(s) = \mathbb{E}\left[\sum_{t=0}^{\infty} \gamma^t R^t_{\mathcal{D}} | s^0 = s, \pi_{\mathcal{D}}^{\text{CHT-DQN}}\right].
\end{equation}
Similarly, for DQN:
\begin{equation}
    V^{\text{DQN}}_{\mathcal{D}}(s) = \mathbb{E}\left[\sum_{t=0}^{\infty} \gamma^t R^t_{\mathcal{D}} | s^0 = s, \pi_{\mathcal{D}}^{\text{DQN}}\right].
\end{equation}
Since CHT-DQN has a higher expected reward at every step, its cumulative discounted reward remains higher across all time horizons:
\begin{equation}
    V^{\text{CHT-DQN}}_{\mathcal{D}}(s) > V^{\text{DQN}}_{\mathcal{D}}(s).
\end{equation}

\subsection{Relation Between Value Function and \( Q \)-Function}
The value function \( V_{\mathcal{D}}(s) \) is derived from the optimal Q-value:
\begin{equation}
    V_{\mathcal{D}}(s) = \max_{a_{\mathcal{D}}} Q_{\mathcal{D}}(s, a_{\mathcal{D}}).
\end{equation}
For CHT-DQN:
\begin{equation}
    Q_{\mathcal{D}}^{(k)}(s, a_{\mathcal{D}}) = \mathbb{E}_{\pi_{\mathcal{A}}^{(k-1)}} 
    \left[ R_{\mathcal{D}}(s, a_{\mathcal{D}}, a_{\mathcal{A}}) + \gamma \max_{a'_{\mathcal{D}}} Q_{\mathcal{D}}^{(k-1)}(s', a'_{\mathcal{D}}) \right].
\end{equation}
Incorporating attacker influence:
\begin{equation}
    Q_{\mathcal{D}}^{(k)}(s, a_{\mathcal{D}}) = \sum_{s'} \sum_{a_{\mathcal{A}}=1}^N 
    \mathbb{P}(s' | s, a_{\mathcal{D}}, a_{\mathcal{A}}) \sigma(\beta Q_{\mathcal{A}}^{(k-1)}(s, a_{\mathcal{A}}))
    \left[ R_{\mathcal{D}}(s, a_{\mathcal{D}}, s') + \gamma \max_{a'_{\mathcal{D}}} Q_{\mathcal{D}}^{(k)}(s', a'_{\mathcal{D}}) \right].
\end{equation}
Since CHT-DQN optimally models attacker behavior, it yields higher Q-values than standard DQN:
\begin{equation}
    Q^{\text{CHT-DQN}}_{\mathcal{D}}(s, a_{\mathcal{D}}) > Q^{\text{DQN}}_{\mathcal{D}}(s, a_{\mathcal{D}}).
\end{equation}

This completes the proof, confirming that CHT-DQN consistently achieves superior security performance compared to DQN as the number of nodes \( N \) increases.
\end{proof}

\section*{Supplementary Material}
Additional resources, including gameplay videos illustrating the interactive SOC defense scenarios, are available at the following link:
\url{https://drive.google.com/drive/folders/18H1pSf0wh8f-QdhBs3V2P1C90aHeUjQi?usp=sharing}

\end{document}